\documentclass[11pt]{article}

\usepackage[margin=1in]{geometry}
\usepackage[T1]{fontenc}
\usepackage{lmodern}
\usepackage{amsmath,amssymb,amsthm,mathrsfs}
\usepackage{graphicx}
\usepackage{hyperref}
\usepackage[numbers]{natbib}
\usepackage{tikz}

\title{An autocomplete model for elections}
\author{Moon Duchin\\University of Chicago\\Chicago, IL, USA\\\texttt{mduchin@uchicago.edu}
\and Kristopher Tapp\\Saint Joseph's University\\Philadelphia, PA, USA\\\texttt{ktapp@sju.edu}}
\date{}

\newcommand{\paperkeywords}{Markov models, social choice, rankings}

\usepackage{enumitem,tikz-cd,xspace,placeins,nicefrac,latexcolors,colortbl,cleveref}
\usetikzlibrary{arrows.meta}
\newtheorem{theorem}{Theorem}

\newtheorem{proposition}[theorem]{Proposition}

\newtheorem{corollary}[theorem]{Corollary}

\newtheorem{definition}[theorem]{Definition}
\newtheorem{example}[theorem]{Example}
\newcommand{\restatedpropname}{}
\newtheorem*{restatedpropinner}{\restatedpropname}

\newcommand{\restatedcorname}{}
\newtheorem*{restatedcorinner}{\restatedcorname}

\theoremstyle{plain}
\newtheorem{remark}{Remark}

\newcommand{\M}{\mathcal{M}}
\newcommand{\B}{\mathcal{B}}
\newcommand{\C}{\mathcal{C}}
\newcommand{\A}{\mathcal{A}}
\newcommand{\tP}{\tilde{P}}
\newcommand{\tS}{\tilde{\sigma}}

\newcommand{\cyc}{\mathop{cycl}}
\renewcommand{\P}{\mathbb{P}}

\newcommand{\SARW}{\textsc{SARW}\xspace}
\newcommand{\LERW}{\textsc{LERW}\xspace}
\newcommand{\REJ}{\textsc{Reject}\xspace}
\newcommand{\KeepFirst}{\textsc{KeepFirst}\xspace}
\newcommand{\MaxEntropy}{\textsc{MaxEntropy}\xspace}
\newcommand{\KeepSmart}{\textsc{KeepSmart}\xspace}
\newcommand{\KeepRandom}{\textsc{KeepRandom}\xspace}
\newcommand{\Terminate}{\textsc{Terminate}\xspace}

\newcommand{\dwass}{d_{\hbox{\scriptsize Wass}}}
\newcommand{\dd}{D}

\begin{document}


\pagestyle{plain}


\maketitle

\begin{abstract}
We consider two linked problems for ranked choice elections: building a compact summary of a preference profile $P$, and generating synthetic profiles that resemble $P$ in key ways. We introduce a Markov model based on which candidate names follow which on observed ballots and record this in a {\em succession matrix} on candidates together with a start/stop symbol. This provides the basis for new ``autocomplete'' methods for generating synthetic profiles, with ballot length that can be tuned to a target. We compare succession-based generation to parallel use of summaries built from positional frequencies and head-to-head comparisons. To execute this comparison, we introduce new descriptive statistics of profiles that allow new tests of whether synthetic profiles are realistic. In experiments on Scottish ranked-choice profiles, succession-based generation captures held-out ballot structure well, performs especially strongly when completing partial ballots, and often preserves election outcomes, outperforming the alternatives on all counts. Finally, we study the dynamics of iterating the train-and-generate process.
\end{abstract}



\section{Introduction}

In a ranked choice election, a {\em preference profile} $P$ is a record of the ballots, i.e., the (complete and partial) rankings by voters of candidates.  A system of election can then be conceived as a voting rule $f$ that carries $P$ to an outcome, whether a single winner, a winner set, or an output ranking. 

We consider two closely related problems:  
\begin{itemize}
    \item Building a compact summary of $P$; \quad and 
    \item Generating  profiles that are similar to $P$ in salient ways.
\end{itemize}
The two are linked: a summary can be judged by how well the profiles reconstructed from it resemble $P$.

The generation of synthetic preferences through probabilistic models, sometimes called {\em statistical cultures}, has received considerable attention in computational social choice theory,
though classically the roster of generative models is quite limited.  
The great bulk of examples found in the literature employ impartial culture (IC), impartial anonymous culture (IAC), Mallows and Urn models, spatial models, and a simple bootstrap (resampling ballots from an observed election).  
More recently, some authors have advocated for the use of models with a lineage in statistical ranking literature, such as Plackett-Luce and Bradley-Terry models, which will be further discussed below.
Except for the bootstrap, all of these options generally output complete rankings, although extensions to incomplete rankings have been considered.  IC and IAC are unparametrized; for the rest, the only connection to a given $P$ is through a handful of learned parameters, such as candidate utilities.

Rather than posit latent parameters governing an underlying preference distribution, we directly train a simple Markov model from empirical succession frequencies, summarizing a profile $P$ in a form that supports robust and flexible reconstruction.  We introduce a {\em succession matrix}, which records how frequently each candidate immediately follows each other candidate on voters' ballots.  This matrix has $m^2+m$ edge weights for $m$ candidates and defines a random walk whose sample paths can be transformed into synthetic ballots of varying lengths.  
We introduce several methods for converting sample paths of the random walk into valid ballots, creating synthetic profiles from a compact model trained on $P$ by simple counting.

We compare our model to parallel approaches based on two other kinds of data from $P$:  a {\em head-to-head} matrix (based on pairwise comparison of candidates) and a {\em position} matrix (based on the frequency with which candidate $i$ is ranked in position $j$).  With an appropriate formulation, these are also means of extracting $O(m^2)$ parameters that admit analogous forms of sequential ballot generation.  We show that---somewhat surprisingly considering the field's emphasis on tournaments and positional scoring---succession data outperforms both head-to-head and positional data across many tests, including reproducing held-out ballot types, completing partial ballots, and preserving STV winner sets.

Completing partial ballots is a particularly practical application: ballot length is capped in many jurisdictions, making it impossible for voters to record complete preferences, and our model can build counterfactual complete ballots based on a profile of partial ballots.

\subsection{Related work}
Some of the practical motivation for this work comes from the experimental side of computational social choice. Numerical studies of voting rules need profiles as input, and the recent {\em Guide to Numerical Experiments on Elections} \citep{boehmer2024guide} documents how heavily the field relies on the small roster of statistical cultures described above, which mainly produce complete rankings from a handful of parameters but are not calibrated to real-world elections. Indeed, Faliszewski et al.~\cite{faliszewski2023diversity} even observe that the field lacks descriptive statistics of profiles, which calibration would require.\footnote{The authors note that existing indices collapse, ``somewhat disappointingly,'' to a single kind, measuring the average margin of victory.} This paper contributes new standards for realism (such as quantifying coverage, novelty, and the shape of the ``long tail'' of rare ballots) and evaluates methods against them. The ability to extend partial ballots is a second application: methods for extending ballots will immediately make real elections usable with algorithms that assume complete ballots.\footnote{Previous work sometimes simply restricts to complete ballots, but those make up a small share of votes even when full rankings are allowed (about 11.6\% of cast ballots in our primary dataset).}

There is a rich literature on probability distributions on rankings. The Plackett--Luce model~\cite{Luce_1959, Plackett_1975} uses candidate weights for individual draws: the ballots are built via draws without replacement with respect to those weights. The Mallows model~\cite{Mallows_1957} assigns each permutation a probability depending on its distance (usually via Kendall tau) to a central or reference permutation. Bradley--Terry models~\cite{Bradley_Terry_1952,Mallows_1957} are paired-comparison models in which all pairwise probabilities come from a single vector of candidate weights.\footnote{Overviews of ranking models can be found in~\cite{Marden_1995,Critchlow_Fligner_Verducci_1991}, which categorize them as paired-comparison, distance-based~\cite{Fligner_Verducci_1986}, multi-stage~\cite{Fligner_Verducci_1988}, and order-statistics models. Bradley--Terry is the special case of the general paired-comparison model of Babington Smith~\cite{babingtonsmith1950discussion}, which generates complete rankings by making all pairwise comparisons independently and rejecting inconsistent outcomes; Davidson~\cite{Davidson_1970} extended Bradley--Terry to accommodate ties, an extension incorporated into the Babington--Smith model in~\cite{Sawadogo_Dossou-Gbete_Lafon_2017}.}

Many statistics authors have fit Plackett--Luce and Mallows models to ranking data~\cite{Lu_Boutilier_2014,Beckett_1993,Busse_Orbanz_Buhmann_2007,Liu_et_al_2019, Hajek_Oh_Xu_2014}, sometimes on election data, notably the 1980 American Psychological Association presidential election~\cite{diaconis1989spectral,Marden_1995} and the 2002 Irish constituency contests~\cite{gormley,gormley2008exploring,gormley2014mixed,Lu_Boutilier_2014,nguyen2023plackettluce}. In all of this work, the fitted distribution lives on complete permutations, so a partial ballot is treated as a censored observation. Ballot length must be supplied from outside in order to output short ballots, and then is typically used independently of the ranking \cite{hoffman2025truncation}.
Awadelkarim and Ugander  persuasively argue that this discards information and distorts downstream simulation, motivating them to fit Plackett–Luce models augmented with a stop token that, they find, needs fitted weights at each rank position \cite{awadelkarim2024topk}.

Standard applications of the ranking models reviewed above typically require fitting model parameters to the observed data; the learned parameters do not necessarily have intuitive interpretations, and the learning or sampling is sometimes computationally costly.\footnote{For example, it is expensive to learn the reference permutations of a mixed Mallows model, though sampling is then fast via the ``repeated insertion method'' of~\cite{Doignon_Pekec_Regenwetter_2004}; a naive application of the Babington--Smith model would cheaply learn head-to-head data from the profile, but then sample via rejection, which becomes intractably slow as the number of candidates grows.}

Our approach is to use a profile $P$ to learn behavior for the candidates together with a start/stop symbol $\star$, from empirical frequencies without fitting utilities (i.e., non-parametrically). The main autocomplete family of models simply bases stopping on successions $i\star$, treated just like other successions $ij$.\footnote{The family includes a maximum-entropy method of generating a synthetic profile, the rest being random-walk methods; the max-entropy construction applies the foundational work of Jaynes from the 1950s~\cite{jaynes1957information}.} 
In parallel, a positional model considers the frequency of observing a stop $\star$ at rank $i$, and a head-to-head model makes use of the pairwise standing of $\star$ with other candidates (considering ranked candidates as above $\star$ and unranked as below). 
This gives a unified treatment to the handling of ballot length and offers novel insights about stopping conditions---notably that reconstituting random walk output into a valid ballot systematically shortens it, which we quantify and then correct by tuning (\S\ref{sec:properties}-\ref{sec:tuning-length}). Even without length tuning, the succession-based model is the most successful by several measures below.

This compares interestingly to a
Markov chain approach from the computer science literature on rank aggregation:  Dwork et al.~\cite{dwork2001rank} build chains whose transitions go from candidates to those ranked higher, not necessarily successive, and they read a consensus order for the preference profile off the stationary distribution of their chains.  For our chains, the stationary distribution is easily characterized (\Cref{prop:steady_state} says that stationary probability is just the frequency of mentions), and the content lies in the transitions, not the steady state.

\subsection{Contributions}

Our primary contributions are as follows.
\begin{enumerate}
    \item We introduce the succession matrix $M$ as a compact summary of a profile $P$, alongside a family of ``autocomplete'' methods to generate new profiles from $M$ via random walks on candidate names (\S\ref{sec:succession}). 
    We contrast this with the use of similar matrices based on positional or head-to-head summary data (\S\ref{sec:comparing}). 
    A start/stop symbol $\star$ gives a uniform treatment of ballot length across the three, and allows users to tune the length of the  outputs (Theorem~\ref{thm:stopping-weights}, Proposition~\ref{P:target_length}).

\item We propose new descriptive statistics of profiles and new standards for assessing whether synthetic profiles are realistic. On the descriptive side, these include stopping weights (for each candidate, the share of ballots ranking them that end with them; Definition~\ref{def:stopping}), a cyclicity measure that vanishes exactly on ordered profiles (Definition~\ref{D:cyclicity}), and the frequency spectrum of ballot types, which captures the shape of the ``long tail'' of rare ballots (Table~\ref{tab:spectrum-pentland}). For comparing generated profiles to real ones, we measure coverage of held-out ballot types and novelty with respect to training data (Figure~\ref{fig:train-test}).
    
    \item Trained on profiles from hundreds of Scottish local government contests, autocomplete recovers held-out ballot types about as well as resampling ballots (bootstrapping) from each $P$ itself, while producing novelty rates close to those of the true held-out data. It is well calibrated on the long tail of ballot types and rarely ``hallucinates'' (\S\ref{subsec:fidelity}).
    \item Trained only on partial ballots, autocomplete has the strongest performance among our compared methods on completing to full ballots; it is the only one that beats a benchmark that samples a short-ballot prefix and completes it uniformly at random (\S\ref{subsec:completion}).
    \item We often want to know how sensitive algorithmic output is to perturbations of the input---in this case, the ranking profile. Generated profiles provide realistic perturbations of $P$; as a case study, we show that succession-based generation yields the original STV winner set more often than positional or head-to-head generation (\S\ref{subsec:stv}).
\end{enumerate}
Finally, we consider what would happen if we iterate the train-and-generate process repeatedly.  We show that the limit points of this process are acyclic, which induces a partial order on candidates (Theorem~\ref{thm:iteration}).  We characterize the limit exactly for two candidates (Proposition~\ref{prop:twocand}) and give a general heuristic.

Replication code for this work is available at \url{github.com/mggg/autocomplete-replication}.

\section{Preliminaries}\label{sec:prelims}

\paragraph{Profiles}
Let $\C=\{A_1,\dots,A_m\}$ be a set of $m$ candidates. A {\em ballot} $\sigma$ is a ranking of some or all of the candidates; it is {\em complete} if it ranks all $m$ and {\em partial} (or top-truncated) otherwise, and its {\em length} is the number of candidates it ranks. We write $i\in\sigma$ to mean that candidate $i$ is ranked on $\sigma$. A ballot of length one is a {\em bullet vote}. A preference profile $P$ is a multiset of ballots, and a voting rule $f$ carries $P$ to an outcome, as described in the introduction. We write $\mu_i$ for the portion of ballots in $P$ that rank candidate $i$ (its frequency of {\em mentions}) and $L$ for the average ballot length in $P$. A profile is {\em ordered} if there is a (possibly not unique) total order of the candidates so that every ballot respects that order.

\paragraph{Weighted graphs}
For a weighted directed graph whose vertices include the candidates, a {\em candidate cycle} is a positive-weight directed cycle passing only through candidates. The graph, or its weighted adjacency matrix, is {\em acyclic} if it has no candidate cycle.

\paragraph{Markov chains}
Recall that a Markov chain on a finite state space is a random walk on $n$ states that is specified by an $n\times n$ transition matrix $M=(m_{ij})$, where $m_{ij}$ is the probability of moving from state $i$ to state $j$, making $M$ row-stochastic.  A sequence of states drawn under proposals from the chain is a {\em sample path}. A probability (column) vector $\pi$ is {\em stationary} if $M^T\pi=\pi$. The chain is {\em irreducible} if every state can be reached from every other, {\em aperiodic} if the return times to each state are relatively prime, and {\em ergodic} if it is both. By the fundamental theorem of Markov chains, as long as the chain is ergodic, $\pi$ exists, is unique, and is the attracting point for all other probability vectors; see~\cite{levin2017markov}.

\section{Succession data and autocomplete}\label{sec:succession}

For a profile $P$ based on $m$ candidates, we build a state space $\Omega=\Omega(P)$ as a weighted directed graph on $m+1$ nodes: one node for each member of the candidate set $\C=\{A_1,...,A_m\}$, plus an additional ``start/stop'' node that we denote by $\star$.  Let $\C^\star:=\C\cup\{\star\}$ denote the extended candidate set.

The directed edge $(i,j)$ has the weight 
$$m_{ij}=\frac{\#\{ \sigma\in P :  ij\in \sigma\}}
{\# \{\sigma\in P: i \in \sigma\}},$$
where $ij\in \sigma$ means that $i$ is ranked immediately above $j$ on the ballot $\sigma$, and $i\in \sigma$ means that $i$ is ranked.
In particular, $\star i,i\star\in \sigma$ mean that $i$ is the first- or last-ranked candidate on the ballot, respectively.
Each pair of distinct nodes has edges both ways, so the state space has $2\cdot\binom{m+1}{2}=m^2+m$ edge weights.
The normalization ensures that every vertex has weighted out-degree $1$; in network science, this is sometimes called the {\em out-strength} of a vertex.

\begin{definition}[Succession and autocomplete]
For a preference profile $P$ on $m$ candidates in which every candidate is ranked by at least one voter, the construction above gives a weighted, directed {\em succession graph} $\Omega(P)$, with weighted adjacency matrix $M=M(P)=(m_{ij})$  called the {\em succession matrix} of $P$.  The family of methods to reconstitute profiles from $M$ is collectively called the {\em autocomplete model}.
\end{definition}

\begin{example}
A 3-candidate example showing $P$, $M$, $\Omega$:

\begin{center}

\begin{minipage}{0.30\textwidth}
\centering
\[
\begin{array}{cccc}
\times 54 & \times 12 & \times 36 & \times 18\\
A & C & C & A \\
C &   & B & B \\
B &   & A & C
\end{array}
\]
\end{minipage}
\hspace{0.08\textwidth}
\begin{minipage}{0.50\textwidth}
\centering
{\small
\[
\bordermatrix{
 & \star & A & B & C \cr
\star & 0 & \nicefrac{3}{5} & 0 & \nicefrac{2}{5} \cr
A & \nicefrac{1}{3} & 0 & \nicefrac{1}{6} & \nicefrac{1}{2} \cr
B & \nicefrac{1}{2} & \nicefrac{1}{3} & 0 & \nicefrac{1}{6} \cr
C & \nicefrac{1}{4} & 0 & \nicefrac{3}{4} & 0 \cr
}
\]
}
\end{minipage}

\vspace{0.8cm}

\begin{tikzpicture}[
    scale=1.15,
    >={Stealth[length=2.2mm]},
    vtx/.style={
        circle,
        fill,
        inner sep=0pt,
        minimum size=4.5pt
    },
    vlab/.style={
        font=\normalsize,
        inner sep=2pt
    },
    wt/.style={
        font=\small,
        inner sep=1pt,
        fill=white
    },
    edge/.style={
        ->,
        semithick,
        shorten >=2pt,
        shorten <=2pt
    },
]

  \node[vtx] (A) at (2,4) {};
  \node[vlab, above] at (2,4.1) {$A$};

  \node[vtx] (B) at (4,0.536) {};
  \node[vlab, below right] at (B) {$B$};

  \node[vtx] (C) at (0,0.536) {};
  \node[vlab, below left] at (C) {$C$};

  \node[vtx] (S) at (2,1.691) {};
  \node[vlab, below] at (S) {\Large $\star$};

  \draw[edge] (S) to[bend left=13]
      node[wt, pos=.5, xshift=-2pt]
      {$\nicefrac{3}{5}$} (A);

  \draw[edge] (A) to[bend left=13]
      node[wt, pos=.5, yshift=-2pt]
      {$\nicefrac{1}{3}$} (S);

  \draw[edge] (S) to[bend left=13]
      node[wt, pos=.5]
      {$\nicefrac{2}{5}$} (C);

  \draw[edge] (C) to[bend left=13]
      node[wt, pos=.5, yshift=2pt]
      {$\nicefrac{1}{4}$} (S);

  \draw[edge] (A) to[bend left=13]
      node[wt, pos=.5]
      {$\nicefrac{1}{6}$} (B);

  \draw[edge] (B) to[bend left=13]
      node[wt, pos=.5]
      {$\nicefrac{1}{3}$} (A);

  \draw[edge] (B) to[bend left=13]
      node[wt, pos=.5]
      {$\nicefrac{1}{6}$} (C);

  \draw[edge] (C) to[bend left=13]
      node[wt, pos=.5]
      {$\nicefrac{3}{4}$} (B);

  \draw[edge] (B) to
      node[wt, pos=.5]
      {$\nicefrac{1}{2}$} (S);

  \draw[edge] (A) to[bend right=8]
      node[wt, pos=.5]
      {$\nicefrac{1}{2}$} (C);

\end{tikzpicture}

\end{center}
\end{example}

$M$ is row-stochastic by construction, and it is the transition matrix of a random walk (Markov chain) on $\Omega$.
We can use $M$ to generate new preference profiles as follows.  
Beginning at $\star$, perform a random walk on $\Omega$ with the weights as transition probabilities.  This produces an infinite sequence of states
$\star, A_{t_1}, A_{t_2}, \dots$
which is a sample path of the random walk.
Up to this point, this closely resembles the idea introduced by Claude Shannon in the early 1950s, using $n$-grams (in this case, 2-grams) to generate plausible strings according to the succession probabilities in a training corpus \cite{Shannon1951}.
By the 1980s this idea had been implemented for predictive text, earning it the name ``autocomplete'' \cite{DarraghWitten1991}.

For our application, we can cut the generated string at instances of $\star$ to obtain {\em raw ballots}, which are strings of candidate names of arbitrary length with possible repetition.  
We will consider several possible methods to adapt the sample path to a collection of ballots, i.e., a synthetic preference profile $\tP$.
\begin{enumerate}
    \item\label{m:reject} (\REJ) Generate raw ballots as above, and keep only the valid ones (rejecting those with duplication).
    \item\label{m:term} (\Terminate) Terminate each ballot the first time a candidate name appears twice.
    \item\label{m:keepfirst} (\KeepFirst) Only keep the first occurrence of each candidate name.\footnote{Several other versions were considered, including  \KeepRandom, which keeps a uniformly random instance of each repeated candidate name, and \KeepSmart, which uses more information about $P$ (like the Borda counts or the head-to-head win/loss/tie portions) to probabilistically decide which of the possible resolutions of a raw ballot to keep. None of these clearly outperforms the simpler \KeepFirst.} 
    \item\label{m:LERW} (\LERW) Processing from left to right, when a second occurrence of any $A_j$ is encountered, delete the entire substring between them.  This is loop-erased random walk on $\Omega$.
\end{enumerate}

For example, if the candidates are $\{A,B,C\}$ and the sample path is 
$$\star A B A C \star A C B \star C B A C B A B \star ...,$$ then the ballots under the processing rules are as follows.
\begin{enumerate}
\item $(ACB)$
\item $(AB), (ACB), (CBA)$
\item $(ABC), (ACB), (CBA)$
\item $(AC), (ACB), (CB)$
\end{enumerate}

Additional alternatives use a self-avoiding random walk or generate ballots from a fitted distribution.
\begin{enumerate}[resume]
    \item\label{m:SARW} (\SARW) Perform a self-avoiding random walk, rejecting any proposed transition that repeats a name already seen on the ballot, ending when $\star$ is reached or when there are no positive-weight edges to unvisited candidates.
    \item\label{m:maxent} (\MaxEntropy) Sample ballots from the maximum entropy distribution with succession probabilities from $M$. (See Appendix~\ref{sec:maxentropy}.)
\end{enumerate}
``Maximum entropy'' roughly means that it can be viewed as the least structured distribution consistent with the given succession probabilities; it introduces no additional relations among ballots beyond what is required to reproduce the succession matrix.   \MaxEntropy provides a tradeoff:  it  outperforms the other methods at generating synthetic profiles close to the training profiles, but it is computationally costly  and becomes intractable even for moderate numbers of candidates, which the other methods can easily handle.

Below, we will rely on a dataset of 1070  Scottish local government elections in which voters ranked 3-14 candidates to elect winners by single transferable vote (STV).  
The election data was originally assembled by David McCune \cite{mccune2024monotonicity} and we use the datasets as published by MGGG \cite{mggg_scot}.
In order to provide the most interpretable comparisons, we will restrict most of the experiments to the number that occurs most often: $m=7$ candidates, which occurs in 282 contests.

\section{Properties}\label{sec:properties}
\subsection{Steady state}
For a succession matrix $M=(m_{ij})$, let $\pi$ be the steady-state vector.  

\begin{remark}[Ergodicity]
As long as every candidate is mentioned in $P$, the Markov chain is irreducible.
As long as there are any ballots of length $1$ and of length $2$ in $P$, the Markov chain is aperiodic.  As recalled above, these conditions constitute ergodicity for the chain and suffice for the existence and uniqueness of $\pi$.
\end{remark}

Recall  that $\mu_i$ is the mentions count for $i$, i.e., the portion of ballots in $P$ that rank candidate $i$, with the convention that $\mu_\star=1$, and  $L$ is the average ballot length in $P$.

\begin{proposition}\label{prop:steady_state}(Description of steady state)
The probability vector $\pi=\frac{\mu}{L+1}$ is stationary for $M$; that is, each candidate's stationary weight is proportional to their frequency of mentions.
\end{proposition}

Thus, $L$ and $\mu$ can be read off of $M$ 
via $L=\nicefrac 1{\pi_\star}-1$ and $\mu_i=\nicefrac{\pi_i}{\pi_\star}$. 
Note that where $m_{ij}$ is the conditional frequency of $ij$ given $i$, we can also write $y_{ij}$ for the absolute frequency of the substring $ij$, so that $m_{ij}=\nicefrac{y_{ij}}{\mu_i}$.

\begin{proof}
We must confirm that $M^T\pi=\pi$, which amounts to verifying that 
$\sum_{j\in\C^\star} m_{ji}\pi_j=\pi_i$ for all $i$.
We compute:
$$\sum_{j\in\C^\star} m_{ji}\pi_j=\sum_{j\in\C^\star} \frac{y_{ji}}{\mu_j}\cdot \frac{\mu_j}{L+1} = \frac{1}{L+1}\sum_{j\in\C^\star} y_{ji} = \frac{1}{L+1}\mu_i=\pi_i.\qedhere$$
\end{proof}

\begin{remark}\label{rem:shrink}
The steady-state weight $\pi_\star$ of the start/stop state $\star$ is just $1/(L+1)$, where $L$ is the average ballot length.  This means that the random walker is at $\star$, on average, every $L+1$ steps, so the expected length of a raw ballot is $L$. 
\end{remark}

\begin{remark}
If all ballots in the profile are complete rankings (i.e., length $m$), then the uniform distribution is stationary.  
\end{remark}

\subsection{Generated ballots and profiles}

Let $P$ be the original profile, $M=(m_{ij})$ its succession matrix.  We can write $\tS$ for the distribution on ballots obtained by $M$ under any of the methods described above.
We can then write $\tP$ for the distribution on random profiles obtained by drawing from $\tS$ to reach the same number of ballots present in $P$.

One simple observation is that, for nearly all methods, the reconstitution process preserves first-place vote frequencies from $P$.  (The exception is \REJ, where ballots with certain first-place choices might be rejected at different rates.) However, other key properties of $P$ that are preserved by raw ballots---like the last-place vote frequencies or the expected ballot length, or even the relative succession probabilities---will in general be altered in the processing from raw to valid ballots.  

\begin{definition}[Stopping weight, circulation]\label{def:stopping}
The {\em stopping weight} of candidate $i$ is $T_i=m_{i\star}$; out of ballots that mention $i$, this is the share of ballots in $P$ that terminate with $i$.
A profile $P$ is called a {\em full circulation profile} if every candidate is sometimes first, sometimes last, and is part of a candidate cycle in $\Omega(P)$.
\end{definition}

It is straightforward to verify that $P$ is ordered if and only if $M$ is acyclic; furthermore, if $M$ is acyclic, then all ballot generation methods are the same, and $\tP$ is ordered with probability $1$.

To measure how far a profile is from being ordered, we use the following statistic.

\begin{definition}\label{D:cyclicity}
For $0<\alpha<1$, the {\em $\alpha$-cyclicity} of the succession matrix $M$ is
\[
\cyc(M,\alpha) = -\log\det(I-\alpha\widehat M) = \sum_{k\geq 1} \frac{\alpha^k}{k}\,\mathrm{trace}(\widehat M^k),
\]
where the $m\times m$ matrix $\widehat M$ denotes the candidate portion of $M$ obtained by removing the row and column of $\star$.
\end{definition}

The $i$th diagonal entry of $\widehat M^k$ measures the probability that a walk starting at $A_i$ returns to $A_i$ in exactly $k$ steps.  Consequently, $\mathrm{trace}(\widehat M^k)$ is the total weight of all loops (closed walks) of length exactly $k$.  The parameter $\alpha$ functions to discount longer loops.  For any fixed value of $\alpha$, $M$ is acyclic if and only if $\cyc(M,\alpha)=0$, so $\cyc(M(P),\alpha)$ vanishes exactly when $P$ is ordered.

\begin{remark}
Full circulation is an extremely mild condition on real-world elections.  For example, all 1070 Scottish local government elections have this property.  In fact, much stronger properties account for most of this:  there are 7786 candidates across all 1070 elections, and 7767 of them (99.76\%) received bullet votes, ensuring both a first- and a last-place ranking.  Furthermore, there are 25,985 pairs of same-contest candidates  in that dataset, and of those, 25,983 form candidate 2-cycles, with both $ij$ and $ji$ appearing successively on some voters' ballots.  
\end{remark}

\begin{theorem}[Ballot statistics]\label{thm:stopping-weights}
If $P$ is a full circulation profile, then the ballot-length and stopping-weight statistics for generated ballots $\tS$ behave as described in Table~\ref{tab:ballot-stats}, and the inequalities are strict.
\end{theorem}

\begin{table}[htb!]\centering
\begin{tabular}{|c|cc|}
\hline 
     & ballot length & stopping weight\\ \hline \hline 
\REJ  & not preserved & up \\ 
\Terminate  & down & up \\ 
\KeepFirst  & down & up \\ 
\LERW  & down & up \\ 
\SARW  & not preserved & up \\ 
\MaxEntropy  & preserved & preserved \\ \hline 
\end{tabular}
\caption{If $P$ is a full-circulation profile, then the expected ballot length is strictly decreased under \Terminate, \KeepFirst, and \LERW processing rules, while it can depend on $P$ for \REJ and \SARW. Relatedly, the stopping weight is strictly increased for all candidates under these methods. Only the computationally costly \MaxEntropy method is guaranteed to preserve these features of $P$.}\label{tab:ballot-stats}
\end{table}

The basic reason for this result is simple:  raw ballots have expected length $L$, and the \Terminate, \KeepFirst, and \LERW processing rules remove symbols in the case of repetition, shortening the ballot. 
More generally, candidate cycles create repeated symbols, and each random-walk generation rule suppresses or modifies cycles in a way that increases the conditional probability of stopping after a given candidate.  Details can be found in Appendix \ref{sec:ballotstats}.

\section{Tuning ballot length via stopping weights}\label{sec:tuning-length}

The results in Table~\ref{tab:ballot-stats} suggest a natural way to compensate for the ballot-shortening effect of reconstitution: decrease the stopping weights before generating ballots.  We will set this up to target a range of ballot lengths.

Suppose throughout this section that $P$ is a full-circulation profile,
so that $0<T_i<1$ for every candidate $i$.  For $\lambda\geq 0$, define a
modified stopping weight $T_i(\lambda)$ by multiplying the odds of stopping
by $\lambda$:
$$\frac{T_i(\lambda)}{1-T_i(\lambda)} = \lambda\cdot\frac{T_i}{1-T_i}.$$
Equivalently, $T_i(\lambda)=\frac{\lambda T_i}{1-T_i+\lambda T_i}$.
Thus $\lambda=1$ gives the original succession matrix, $\lambda=0$ suppresses
stopping entirely, and $\lambda\to\infty$ forces increasingly early stopping.
In terms of log-odds, this simply shifts every stopping weight by adding $\log \lambda$.

Let $M(\lambda)$ be obtained by replacing each $T_i$ by $T_i(\lambda)$ and
rescaling the remaining entries in row $i$ proportionally:
$$m_{ij}(\lambda)
      = m_{ij}\cdot \frac{1-T_i(\lambda)}{1-T_i},
      \qquad j\neq \star.$$
Thus the relative transition probabilities from $i$ to other candidates
are unchanged.
For a ballot-generation rule $R$, let $L_R(\lambda)$ denote the expected
ballot length obtained by applying $R$ to $M(\lambda)$.  We say that a target
length $\ell$ is attainable for $R$ if $L_R(\lambda)=\ell$ for some $\lambda$,
and write $R_\ell$ for the resulting tuned method.

\begin{proposition}[Target ballot length]
\label{P:target_length}
For autocomplete with a processing rule $$R\in \{\Terminate,\KeepFirst, \SARW \},$$ the function
$\lambda\mapsto L_R(\lambda)$ is continuous and non-increasing.  Moreover, since 
$\lim\limits_{\lambda\to\infty}L_R(\lambda)=1$,   every target length $1< \ell\leq L_R(0)$ is attainable.
\end{proposition}

\begin{proof}
After conditioning on continuation from candidate $i$, the transition
probabilities among candidates are independent of $\lambda$.  Increasing $\lambda$
only increases the probability of stopping at each stage.  Thus the processes can be coupled so that increasing $\lambda$ can only shorten the
generated ballot, establishing monotonicity.  Continuity follows directly from the transition probabilities.

As $\lambda\to\infty$, stopping occurs immediately after the first candidate with probability tending to one, so $L_R(\lambda)\to 1$.  At $\lambda=0$, ordinary
stopping is suppressed completely.
The intermediate value theorem now gives every target length between $1$ and $L_R(0)$.
\end{proof}

$L_{\Terminate}(0)$ has a simple interpretation: the expected number of distinct candidates visited before the first repeat in the candidate-only random walk. Thus suppressing ordinary stopping does not suppress termination caused by cycles. In particular, $\Terminate_L$ exists iff $L\leq L_{\Terminate}(0)$. This condition holds for almost all, but not all, of the Scottish profiles.

The same odds-shifting construction makes sense for \REJ and
\LERW whenever $\lambda>0$, so these methods can also be tuned numerically
to any target length that occurs in the range of $L_R$.  What is
missing is the clean existence result above.  At $\lambda=0$ the raw walk
never terminates, and conditioning on validity for \REJ or applying
loop erasure for \LERW makes the effect of decreasing $\lambda$ less directly
controlled.  In particular, the original mean ballot length $L=L(P)$ need not
belong to the attainable range for \LERW.

\begin{definition} 
Recall that $\KeepFirst_\ell$ and $\SARW_\ell$ denote the result of using \KeepFirst and \SARW respectively with $M(\lambda)$ for $L_R(\lambda)=\ell$.  Special cases are $\KeepFirst_L,\SARW_L$, which match the training profile $P$ in expected ballot length, and $\KeepFirst_m,\SARW_m$, which set stopping weights to zero to generate complete ballots.
\end{definition}

Monotonicity allows a fast iterative algorithm for finding the right scaling factor $\lambda$ to hit a target length $\ell$.  We will show ballots tuned to full length ($\ell=m$) in Figure~\ref{fig:completion} and to the observed average ballot length from an election ($\ell=L(P)$) in Figure~\ref{fig:dist-to-P}.

A theoretically slightly more flexible approach is to tune the $m$ stopping weights individually, rather than adjusting them all uniformly, to approach the stopping weights of $M$; however, this trades simplicity for flexibility.  Instead, see Appendix~\ref{sec:maxentropy} for the development of the \MaxEntropy method, which can roughly be thought of as tuning all of the succession weights, not just the stopping weights, to attain fidelity to the whole succession matrix.

\section{Comparing generative methods}\label{sec:comparing}

To assess the performance of the succession matrix as a profile summary, it will be helpful to make direct comparisons of autocompleted profiles with parallel generative methods that use positional or head-to-head comparison data. Our aim  is not to compare against the best available fitted ranking models, but to compare the information retained by three directly empirical summaries.  To that end, we briefly recap the succession model and build similarly fast and local methods that work from the other two summaries.  All three methods have a similar stopping role for the symbol $\star$, so all three admit similar forms of length tuning.

\paragraph{Succession model} 
From the observed profile $P$, form the succession matrix $M$ on the candidates plus start/stop symbol $\star$. Here $m_{ij}$ is the frequency with which $j$ immediately follows $i$, conditional on $i$ appearing.
Generate a ballot sequentially, starting at $\star$. If the current state is $i$, choose the next state from $\star$ and the as-yet-unranked candidates, with probabilities proportional to the corresponding entries $m_{ij}$. If $\star$ is chosen, stop; otherwise append the chosen candidate and continue.
This is Autocomplete-\SARW.

\paragraph{Position model}
For every rank $r$, form position frequencies  for the candidates together with $\star$: $f_{rj}$ is the frequency with which candidate $j$ appears in position $r$, while
$$f_{r\star}=\P\left( \hbox{\rm ballot stops at rank} ~ r-1 \right).$$

Generate a ballot from top to bottom. At rank $r$, choose from $\star$ and the candidates not yet ranked, with weights $f_{rj}$. Remove already-ranked candidates and renormalize. If $\star$ is chosen, stop; otherwise append the selected candidate and continue.  

\paragraph{H2H model}
From $P$, form the head-to-head (H2H) matrix $H$ on the candidates together with the start/stop symbol $\star$. For candidates $i,j$,
$h_{ij}$ is the share of all ballots on which $i\succ j$, counting a ballot that ranks exactly one of them as a win for the ranked candidate.  Thus $h_{ij}+h_{ji}$ is the share of ballots that decide the pair by ranking at least one of them, and
$h_{j\star}=1-h_{\star j}=\mu_j$, the frequency of mentions.

Choose a uniformly random arrival order of $\C^\star$, and build a list by randomized insertion. When $C$ arrives and
the current list is $A_1\dots,A_k$, give the position after the first $t$ entries the weight
$w_t(C)=h_{A_t,C}\, h_{C, A_{t+1}}$, where $w_0(C)=h_{C,A_1}$ and $w_k(C)=h_{A_k,C}$. Choose an insertion position
proportional to these weights; the generated ballot is the segment of the
final list above $\star$.  Once $\star$ has been placed, the candidates below it are an unordered pool rather than a ranked list, so an arriving candidate sees the gaps of the ranked segment together with a single terminal position below $\star$, of weight $h_{\star C}$.\footnote{The ballot distribution is computed by a forward dynamic program over generator states, averaging over the random arrival order. Note also that only H2H can produce empty ballots; this happens in up to 9\% of generated ballots in some of the experiments, and we drop empty ballots in reported statistics.}
Thus the H2H model is a randomized insertion sort driven by the observed pairwise preferences.

In the completion experiment, stopping is suppressed for every model:
for Succession this means deleting the $\star$ column from $M$, and for H2H it means omitting $\star$ from the arrival order.

Each generator thus assembles ballots sequentially, with local decisions only.

\section{Empirics}\label{sec:empirics}
In this section, we study the properties of profiles $\tP$ which are generated via the methods described above from a training profile $P$.
We will use the Wasserstein (``earth-mover'') distance as one important metric on profiles.  That is, we consider $P$ and $\tP$ as probability mass distributions on the set of possible ballots, which is a metric space with respect to the Spearman footrule extended to partial ballots in~\cite{Duchin_Tapp_2026} as a ``pessimistic Borda distance.'' We define 
$\dd(P,\tP):= \dwass(P,\tP) / \mathop{\rm diam}$, 
where $\mathop{\rm diam}=\lfloor m^2/2\rfloor /2$ is the diameter of the space of complete ballots.
This gives a distance on $m$-candidate profiles that is normalized so that complete-ballot profile comparisons lie in $[0,1]$, and in fact generalizes to a distance between {\em measures} on the ballot set.


\subsection{Structural fidelity}\label{subsec:fidelity}

An election generator trained on a profile $P$ would ideally be able to create new profiles that are similar to $P$ in structural ways. In Figure~\ref{fig:train-test}, we compare generators as follows:  for each of the 282 7-candidate elections in the Scottish dataset, draw 20\% of the ballots at random as training data, regarding the other 80\% as held-out testing data.  Three particular elections from 2017---Pentland Hills (Edinburgh Ward 2), North Ayrshire Ward 7, and Renfrewshire Ward 10---are highlighted in the data.

\begin{figure}[ht]
    \centering
    \includegraphics[width=4.4in]{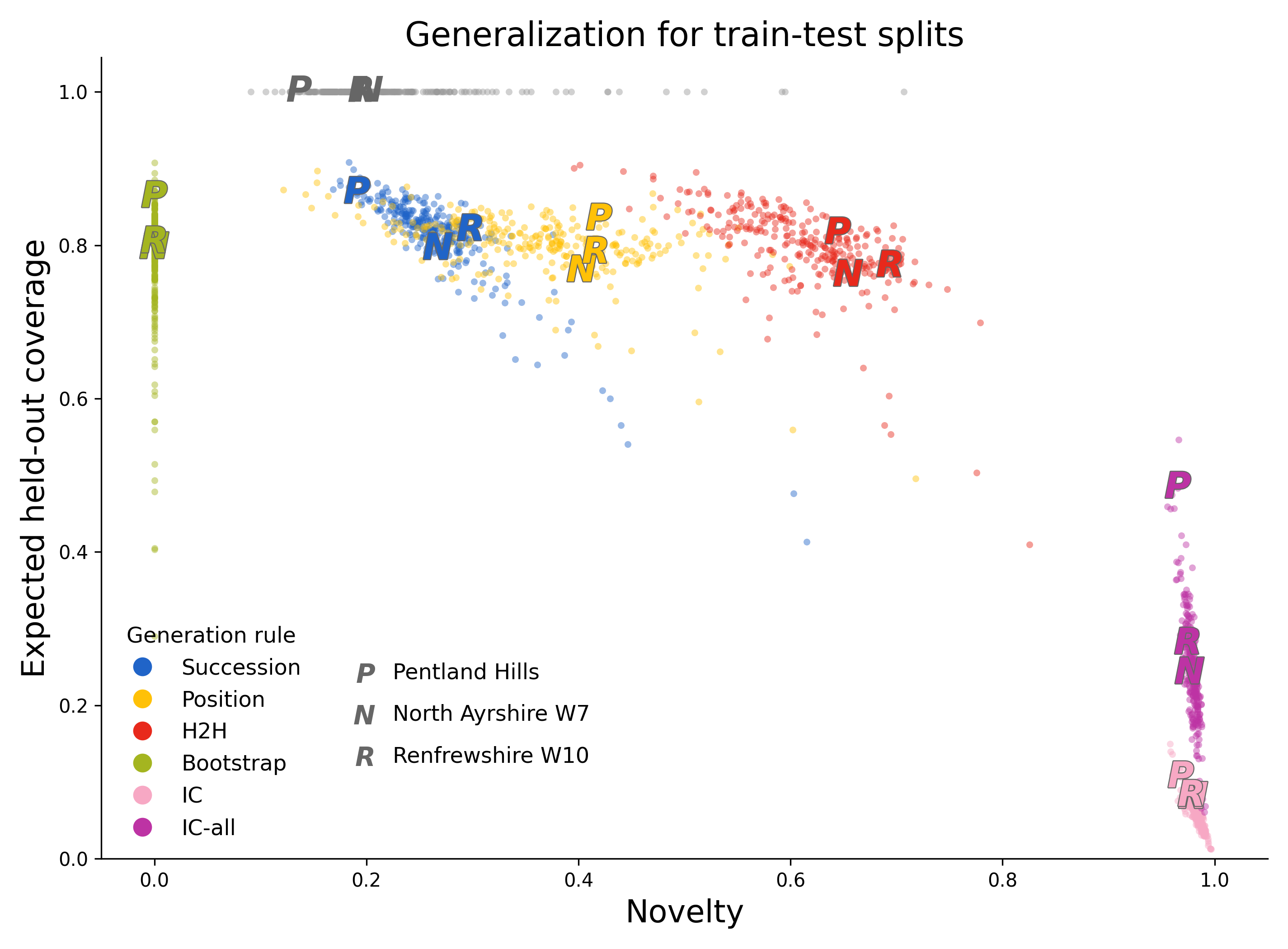}
    \caption{For the 282 7-candidate Scottish elections, we compare the performance of different generative models trained on a subsample of the given profile $P$, using a single 20-80 training-testing split for each. {\em Novelty} and  {\em coverage} are exact expectations with respect to the distribution $\tP$. Three specific elections are highlighted.}
    \label{fig:train-test}
\end{figure}

The {\em coverage} of a distribution $\tP$ is the expected share of the held-out ballots whose type occurs in a generated election of the same size (80\% of the original $P$).  The {\em novelty} of the method is the expected share of generated ballots whose type is absent from the training data.  
The actual held-out sets are plotted in gray and appear at the top:  they necessarily have perfect coverage,  but their rate of novelty with respect to the training data varies.  

In the plot, IC (impartial culture) puts equal weight on all $m!$ complete ballots, while IC-all puts equal weight on all well-formed ballots, whether complete or partial.  Unsurprisingly, these have very high novelty rates but poor performance at recovering the held-out ballots.  Bootstrap refers to resampling from the training data; it can't produce any novelty, by construction.
Of the named elections, Pentland Hills is the largest election of the three (11,315 voters compared to 4733 and 5630), so the 20\% training data is richest and the 80\% generated sample is largest. This explains the better coverage for the $P$ election across all methods.\footnote{Figure~\ref{fig:train-test} uses a single 20-80 split per election, and the novelty and coverage calculations are exact with respect to that split. To assess sensitivity to this choice, we repeated the full experiment on 20 independent splits of the Pentland Hills contest; across models, the split-to-split standard deviations were at most 0.004 for coverage and 0.02 for novelty---roughly an order of magnitude below the between-model differences visible in the figure.}

A realistic generator should comport with the held-out novelty rate. Succession achieves coverage similar to the bootstrap, while matching the novelty rate most closely.

\begin{table}[htb!]\centering
\footnotesize
\def\overcolor{\cellcolor{teal!50}}
\def\undercolor{\cellcolor{lust!50}}
\begin{tabular}{|c||cccccc||c|}
\hline
 & $k=1$ & $k=2$ & $k=3$ & $k=4$ & $k=5$ & $k\ge 6$ & types\\ \hline \hline
expectation  & 598.1 & 163.7 & 76.0 & 47.5 & 29.0 & 183.4 & 1097.7 \\
sample draw  & 582 & 169 & 85 & 38 & 31 & 188 & 1093 \\ \hline \hline 
Succession  & 666.1 & 183.1 & 84.8 & 48.2 & 30.6 & 167.2 & 1180.0 \\
Position  & 638.6 & 155.4 & 71.1 & 42.9 & 29.8 & \overcolor 226.6 & 1164.5 \\
H2H  & \overcolor 2368.7 & \overcolor 576.0 & \overcolor 209.0 & \overcolor 103.0 & \overcolor 60.0 & \overcolor 221.5 & \overcolor 3538.2 \\
Bootstrap  & \undercolor 20.7 & \undercolor 41.7 & \undercolor 56.5 & \overcolor 58.5 & \overcolor 49.9 & \overcolor 242.5 & \undercolor 469.9 \\
IC  & \overcolor 1502.3 & \overcolor 1349.2 & \overcolor 807.7 & \overcolor 362.6 & \overcolor 130.2 & \undercolor 51.8 & \overcolor 4203.7 \\
IC-all  & \overcolor 4675.1 & \overcolor 1544.5 & \overcolor 340.1 & 56.2 & \undercolor 7.4 & \undercolor 0.9 & \overcolor 6624.3 \\
\hline
\end{tabular}
\caption{Frequency spectrum for Pentland Hills: the expected number of ballot types occurring exactly $k=1,2,3,4,5$ times (and six or more times) in a test-size sample, with the expected total number of distinct types. Cells are shaded \colorbox{lust!50}{red}
and \colorbox{teal!50}{teal}
when more than 20\% off from true expectation.}\label{tab:spectrum-pentland}
\end{table}

Table~\ref{tab:spectrum-pentland} compares the methods in terms of accurately capturing the shape of the long tail of the profile $P$:  ballot types that occur a small number of times in the test data.  In this table, the expected number is exact with respect to 80\% random samples of $P$, while the specific sample draw is from the split used in Figure~\ref{fig:train-test}.  Succession is always within 20\% of the correct mass (and Position is similarly successful but somewhat overweights the more common types).

We can also compare the methods in terms of rare-ballot calibration: are there any specific types that are ``hallucinated''---that is, high-confidence predictions with no real-world referent?  Positional generation is particularly hallucination-prone, in interesting ways.  In Pentland Hills, the $1\succ5$ ballot is never observed in either the training or the testing data; it is predicted to occur $11.6$ times when Succession creates profiles at test size, versus $235.9$ times in Positional output and $45.8$ times in H2H output.  This is essentially because it is a party-disrupting ballot, with the pattern Con$\to$SNP (a Conservative candidate followed by one from the Scottish National Party).  Candidate 1 is often first and candidate 5 is often second in $P$, but typically paired with members of their own party.\footnote{Similar party-disrupting ballots $3\succ 6$ in North Ayrshire and $5\succ4\succ2$ in Renfrewshire get an expectation of 32.3 and 19.7 Positional occurrences, respectively, though they do not appear in the elections and are correctly labeled as unlikely by Succession. The situation is not reciprocal:  the most-hallucinated Succession ballot in Pentland is $6\succ 2\succ 7$, but its expected Succession frequency of 14.7 is smaller than its expected Position frequency of 20.4.} Succession data is well-suited to pick up on same-party runs.

\subsection{Completion}\label{subsec:completion}

We test whether patterns learned from short ballots predict the distribution of complete rankings, using the $\ell=m$ setting from \S\ref{sec:tuning-length}. In each of the same 282 seven-candidate elections, we train on the subset of the profile made up of ballots of length at most four. Then we compare generated complete ballots with observed complete ballots using normalized Wasserstein distance.\footnote{For Positional data specifically, we fill in the candidate in rank $r$  on a generated ballot with probabilities proportional to rank-$r$ frequencies; because training data doesn't see lower ranks, we reuse the rank-4 frequencies for ranks 5 and beyond.}

The {\em random completion} benchmark (second column) samples from the short ballots and orders the remaining candidates uniformly at random. The {\em bootstrap benchmark} is shown in gray at the bottom of the figure: if the original election has $N$ complete ballots, draw $N$ with replacement from that set itself and measure the normalized Wasserstein
distance from the sample to the original.

Figure~\ref{fig:completion} shows that only autocomplete improves on random completion: the learned succession patterns provide useful information for predicting complete rankings.

\begin{figure}[ht]
    \centering
    \includegraphics[width=4.4in]{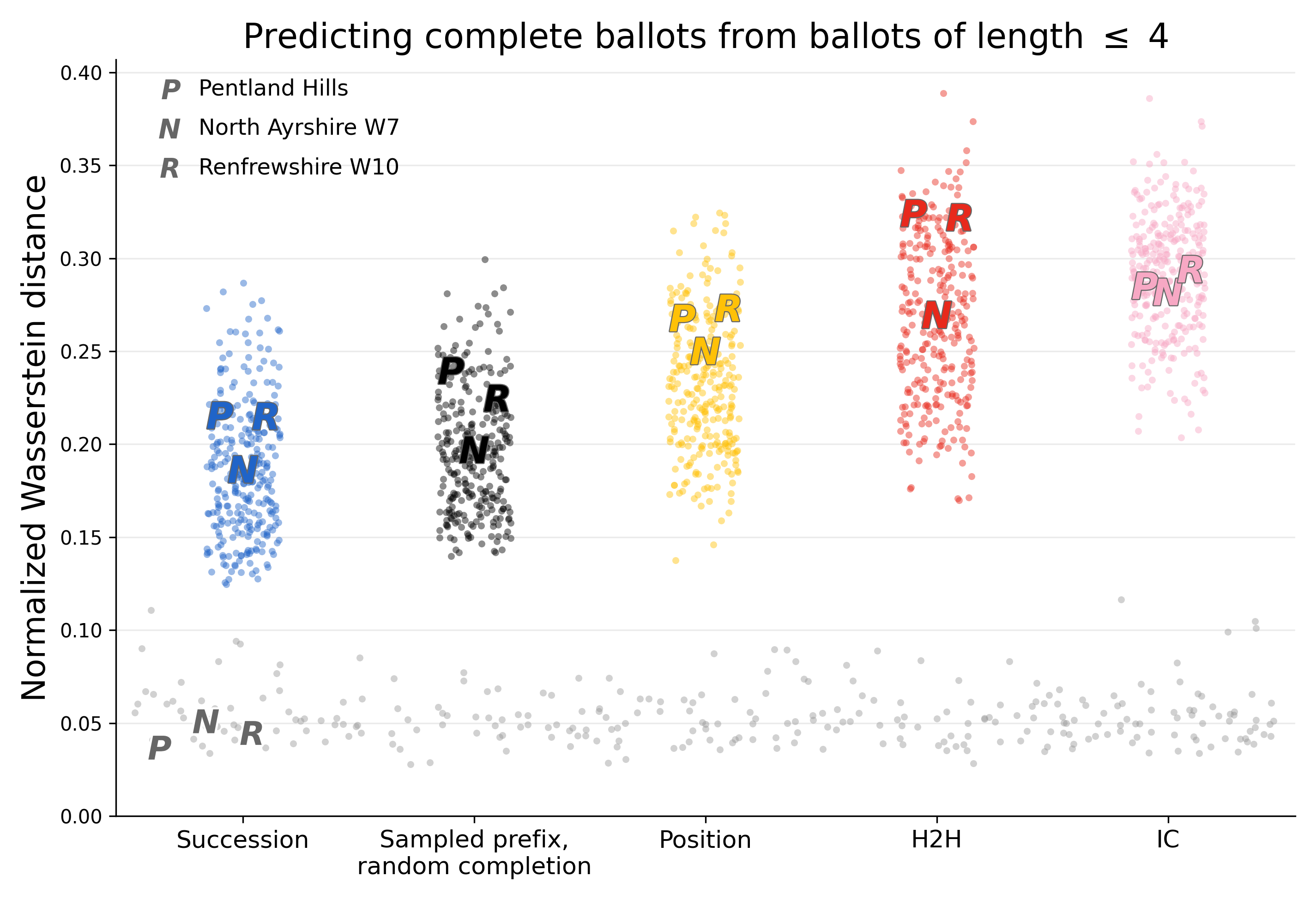}
    \caption{For the same contests, we now compare the performance of the generators at predicting complete ballots ($\ell=m$) when trained only on ballots of length $\le 4$.}
    \label{fig:completion}
\end{figure}

\subsection{STV outcomes}\label{subsec:stv}

One important measurement of fidelity of the synthetic profile $\tP$ to the original profile $P$ is the set of winners.  For each $P$, we train each model on the full profile, generate 10 new profiles of the same size, and compare the new STV winner sets to the original.

\begin{figure}[ht]
    \centering
    \includegraphics[width=4.4in]{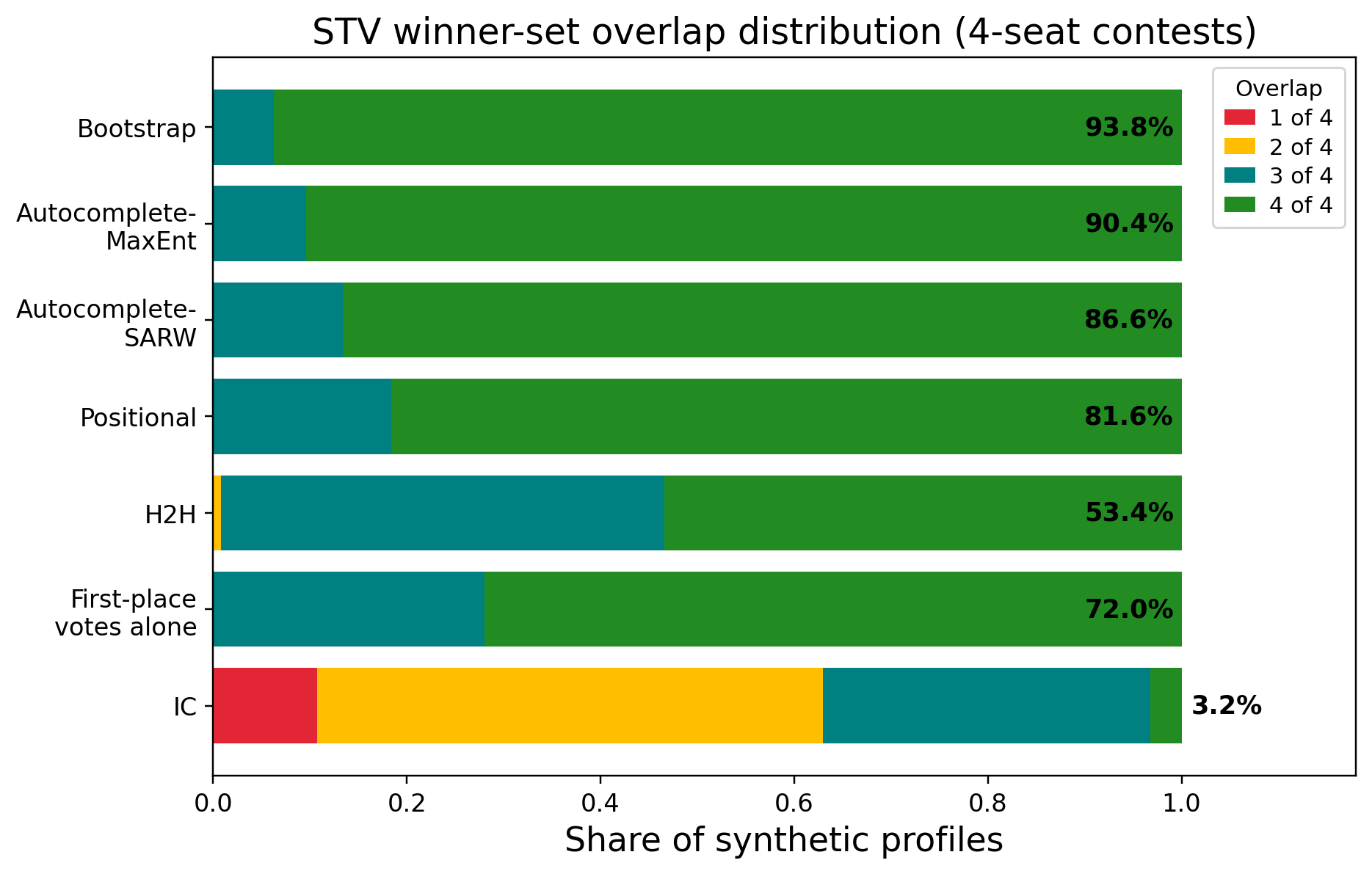}
    \caption{If the reconstituted preferences are fed to an STV voting rule, how much does the winner set $f(\tP)$ overlap with the original $f(P)$?  This figure shows the results for the 125 Scottish contests with 7 candidates and 4 winners. The {\em First-place votes alone} comparator simply takes the four candidates with the most first-place votes from $P$.}
    \label{fig:STV-k4}
\end{figure}

In Figure~\ref{fig:STV-k4}, we see that even resampling from the full original profile changes winners approximately 6 percent of the time, reflecting a sensitivity of STV election itself.  The ``first-place votes alone'' benchmark simply selects the candidates with the most first-place support; since STV is a round-by-round election system that is initialized with first-place votes, this is a crude proxy for the eventual winners.  Positional generation, autocomplete-\SARW, and especially autocomplete-\MaxEntropy all improve on the first-place-vote proxy substantially, with the maximum entropy method performing nearly as well as the bootstrap from the full $P$. By contrast, the head-to-head data yields poor predictions of STV winners.  This is largely because, of the data considered, head-to-head data is the worst for the purpose of reconstructing first-place votes.

Results for the 157 7-candidate contests with 3 winners are very similar; see Figure~\ref{fig:STV-k3} in Appendix~\ref{sec:empirical-details}.

\section{Iteration dynamics}\label{sec:iteration}

We have seen that the succession matrix of a generated $\tP$ need not equal
the original matrix $M$: by Theorem~\ref{thm:stopping-weights}, candidate cycles
create repeated symbols, and each random-walk generation rule suppresses or
modifies cycles in a way that increases the conditional probability of
stopping after a given candidate. To close our analysis, we study the
dynamical system obtained by iterating the train-and-generate operation to
obtain a sequence
$P_0\rightarrow M_0\rightarrow P_1 \rightarrow M_1 \rightarrow\cdots$. Since
the limit behavior depends only on $M_0$, one motivation is to better
understand what information about a preference profile is encoded in the
succession matrix.

Since all of our rules generate profiles via i.i.d. draws from a distribution on ballots, it is cleaner to iterate between transition matrices and distributions on ballots (rather than on profiles).  For this, let $\M$ be the set of transition matrices on $\C^\star=\{\star,A_1,\dots,A_m\}$, $\B$ the set of (partial or complete) nonempty ballots, and $\Delta(\B)$ the set of probability distributions on $\B$ for which each candidate has a positive probability of being mentioned on a ballot.

For a given ballot generation rule $R$, we have maps as follows:
$$\begin{tikzcd}
\Delta(\B)
  \arrow[r, bend left=20, "g"]
&
\M.
  \arrow[l, bend left=20, "f_R"]
\end{tikzcd}$$
Here $g$ carries a probability measure on ballots, $\nu\in\Delta(\B)$, to the corresponding transition matrix $M$, defined just like with profiles; that is $M_{ij}=\P_{\sigma\sim\nu}(ij\in\sigma\mid i\in\sigma)$ is the probability of $ij$ given $i$ for a ballot drawn from $\nu$.  The map in the other direction, $f_R:\M\to\Delta(\B)$, is defined so that $f_R(M)(\sigma)$ is the probability that a ballot produced by $M$ under rule $R$ equals $\sigma$.  We will write $f$ when the rule $R$ is fixed, and $\Phi_R = g\circ f_R:\M\to\M$ for the induced map on matrices.

We consider  $R$ to be one of $\Terminate,\SARW, \LERW, \KeepFirst$. \MaxEntropy is excluded
because its dynamics are trivial: it is faithful to the succession matrix by
construction, so $\Phi_{\MaxEntropy}$ is the identity. 

We call $M$
{\em two-sided} if $m_{\star i}>0$ and $m_{i \star}=T_i>0$ for every
candidate $i$ (meaning that every candidate is sometimes first and sometimes
last); the succession matrix of a full-circulation profile is two-sided.

\begin{theorem}\label{thm:iteration}
Let $M$ be two-sided. Then $\Phi_R(M)$ has the same $\star$ row as $M$, and
$T_i(\Phi_R(M))\geq T_i(M)$ for every candidate $i$, with strict inequality
for every candidate $i$ that lies on a positive-weight candidate cycle in $M$.
Moreover:
\begin{enumerate}
\item $M$ is fixed under iteration (that is, $\Phi_R(M)=M$) if and only if
$M$ is acyclic.
\item Every limit point of the orbit $M_n=\Phi_R^n(M)$ is an acyclic matrix
fixed by $\Phi_R$; equivalently, every profile it generates is ordered. If
$R=\Terminate$, then the orbit has a unique limit.
\end{enumerate}
\end{theorem}

Our experimental evidence (Appendix~\ref{sec:iteration-plots}) is consistent with uniqueness of limits for all $R$  but the proof is specific to $R=\Terminate$.
Full proofs are in Appendix~\ref{sec:iteration-details}. The limit encodes a partial ordering of the candidates; this order can depend on the ballot rule $R$ but examples suggest that discrepancies occur only when a comparison is near equilibrium.

We now briefly present the complete solution for the rules \Terminate and
\SARW in the smallest nontrivial case, when there are $m=2$ candidates, in
which case these two methods coincide. A profile in the $m=2$ case is fully
specified by $(a,b,c,d)$, the share of $(A)$, $(B)$,
$\left(\begin{smallmatrix} A\\ B\end{smallmatrix}\right)$,
$\left(\begin{smallmatrix} B\\ A\end{smallmatrix}\right)$ ballots,
respectively.

\begin{proposition}\label{prop:twocand}
Suppose a profile $P$ with two-sided succession matrix has ballot type frequencies $(a,b,c,d)$ as above, and
let $\alpha=a+c$ and $\beta=b+d$ be the first-place shares for $A$ and $B$,
respectively. Define the comparison function
$\varphi=\alpha^2\frac{c}{1-b} -\beta^2\frac{d}{1-a}$.
Then iterating autocomplete under $R=\Terminate,\SARW$ yields an exact limit:
$$\begin{cases}
(\alpha-s,\beta,s,0),& \varphi>0\\
(\alpha,\beta-t,0,t),& \varphi<0\\
(\alpha,\beta,0,0),& \varphi=0,
\end{cases}$$
for some survival weights $s\le c$, $t\le d$.
\end{proposition}

We see that the first-place weights are preserved at $\alpha$ and $\beta$.
If $P$ has no bullet votes ($a=b=0$) then $\varphi=c^3-d^3$ and the direction
$A\to B$ or the reverse is just determined by majority preference.

For Pentland Hills, North Ayrshire, and Renfrewshire, the cyclicity $\cyc(M_n,1/2)$ becomes very small with all methods after only a few dozen steps. Computing the limits for these profiles lets us observe that the comparisons generalizing $\varphi$ are highly predictive.  Define $\varphi_{ij}=\alpha_i^2 m_{ij}-\alpha_j^2 m_{ji}$ in terms
of the first-place shares $\alpha_i,\alpha_j$ and the transition
probabilities from $M$. The sign of $\varphi_{ij}$ from the original matrix
$M$ predicts whether $ij$ or $ji$ survives in the limit in 57/63 candidate
pairs across the three elections. For Pentland Hills, the limits
(Figure~\ref{fig:limits} in Appendix~\ref{app:3limits}) allow us to see that
the strongest surviving successions are sequentially ordered within parties.

\section{Conclusion}

In this paper,  we propose that the succession matrix of an observed preference profile serves as a compact summary of the rankings that retains enough information to cheaply generate synthetic profiles that are faithful to many of the key properties of $P$. Under several methods of generating ballots from the succession matrix, we establish theoretical properties of the resulting distribution on ballots. Experiments conducted on a significant corpus of real-world ranking data illustrate the compression and regeneration on realistically complex training data, showing that the algorithms operate efficiently at practical scale.

The empirical results suggest that relatively simple local transition information captures a surprising amount of the global structure of a profile. Measured by the descriptive statistics and fidelity tests introduced here, the synthetic profiles tend to do a good job matching testing data on both the novelty (with respect to training data) and coverage of ballot types.  The autocomplete model is well calibrated on rare ballots,  completes partial ballots better than the alternatives, and frequently reproduces the original STV winner sets. Though \MaxEntropy processing of succession data turns in the best measurable performance against many benchmarks, the fast and scalable random-walk-based methods still perform strongly and compare favorably to alternatives drawn from positional and tournament-based data.

Finally, we introduce the dynamical system constructed by iterating the reconstruction process. We show that its limit points are acyclic and characterize them exactly in the base case of two candidates.  In general, an acyclic limit point induces a partial order on candidates, and we give a heuristic for predicting it.

One limitation is that autocomplete is trained on a single profile; in particular, the profiles it generates have the same candidates as $P$, so it cannot produce elections with a different candidate pool.
Pooling succession data across many elections, for instance by leveraging party structure or otherwise allowing candidate pools to vary, is a natural direction for future work.

\section*{Acknowledgments}The authors are pleased to thank Peter Rock, Piotr Faliszewski, Sarah Cannon, Jennifer Wang, Mala Mukundan, Jack Gibson, and Zach Glaser for useful conversations and overlapping collaborations.  

\bibliographystyle{plain}
\bibliography{bibliography}

@book{levin2017markov,
  author    = {Levin, David A. and Peres, Yuval},
  title     = {Markov Chains and Mixing Times},
  edition   = {2nd},
  note      = {With contributions by Elizabeth L. Wilmer},
  publisher = {American Mathematical Society},
  year      = {2017}
}

@misc{mggg_scot,
  author       = {{MGGG}},
  title        = {Scottish local government election data ({scot-elex})},
  howpublished = {\url{https://github.com/mggg/scot-elex}}
}

@article{mccune2024monotonicity,
  author  = {McCune, David and Graham-Squire, Adam},
  title   = {Monotonicity anomalies in {S}cottish local government elections},
  journal = {Social Choice and Welfare},
  volume  = {63},
  pages   = {69--101},
  year    = {2024},
  doi     = {10.1007/s00355-024-01522-5}
}

@inproceedings{boehmer2024guide,
  title     = {Guide to Numerical Experiments on Elections in Computational Social Choice},
  author    = {Boehmer, Niclas and Faliszewski, Piotr and Janeczko, {\L}ukasz and
               Kaczmarczyk, Andrzej and Lisowski, Grzegorz and Pierczy{\'n}ski, Grzegorz and
               Rey, Simon and Stolicki, Dariusz and Szufa, Stanis{\l}aw and W{\k{a}}s, Tomasz},
  booktitle = {Proceedings of the Thirty-Third International Joint Conference on
               Artificial Intelligence, {IJCAI-24}},
  publisher = {International Joint Conferences on Artificial Intelligence Organization},
  editor    = {Kate Larson},
  pages     = {7962--7970},
  year      = {2024},
  month     = {8},
  note      = {Survey Track},
  doi       = {10.24963/ijcai.2024/881},
}

@inproceedings{faliszewski2023diversity,
  title     = {Diversity, Agreement, and Polarization in Elections},
  author    = {Faliszewski, Piotr and Kaczmarczyk, Andrzej and Sornat, Krzysztof and
               Szufa, Stanislaw and Was, Tomasz},
  booktitle = {Proceedings of the Thirty-Second International Joint Conference on
               Artificial Intelligence, {IJCAI-23}},
  publisher = {International Joint Conferences on Artificial Intelligence Organization},
  editor    = {Edith Elkind},
  pages     = {2684--2692},
  year      = {2023},
  month     = {8},
  note      = {Main Track},
  doi       = {10.24963/ijcai.2023/299},
}

@inproceedings{dwork2001rank,
  author    = {Dwork, Cynthia and Kumar, Ravi and Naor, Moni and Sivakumar, D.},
  title     = {Rank aggregation methods for the {W}eb},
  booktitle = {Proceedings of the 10th International Conference on World Wide Web (WWW '01)},
  pages     = {613--622},
  year      = {2001},
  publisher = {ACM}
}

@article{gormley2008exploring,
  title={Exploring voting blocs within the Irish electorate: A mixture modeling approach},
  author={Gormley, Isobel Claire and Murphy, Thomas Brendan},
  journal={Journal of the American Statistical Association},
  volume={103},
  number={483},
  pages={1014--1027},
  year={2008},
  publisher={Taylor \& Francis}
}

@article{gormley2014mixed,
  title={Mixed membership models for rank data: Investigating structure in Irish voting data},
  author={Gormley, Isobel Claire and Murphy, Thomas Brendan},
  journal={Handbook of mixed membership models and their applications},
  pages={441--460},
  year={2014},
  publisher={Chapman and Hall/CRC}
}

@InProceedings{gormley,
author="Gormley, Isobel Claire
and Murphy, Thomas Brendan",
editor="Airoldi, Edoardo
and Blei, David M.
and Fienberg, Stephen E.
and Goldenberg, Anna
and Xing, Eric P.
and Zheng, Alice X.",
title="A Latent Space Model for Rank Data",
booktitle="Statistical Network Analysis: Models, Issues, and New Directions",
year="2007",
publisher="Springer Berlin Heidelberg",
address="Berlin, Heidelberg",
pages="90--102",
isbn="978-3-540-73133-7"
}

@article{DarraghWitten1991,
  author  = {Darragh, John J. and Witten, Ian H.},
  title   = {Adaptive Predictive Text Generation and the Reactive Keyboard},
  journal = {Interacting with Computers},
  volume  = {3},
  number  = {1},
  pages   = {27--50},
  year    = {1991},
  month   = apr,
  doi     = {10.1016/0953-5438(91)90004-L}
}

@article{Shannon1951,
  author  = {Shannon, Claude E.},
  title   = {Prediction and Entropy of Printed {English}},
  journal = {Bell System Technical Journal},
  volume  = {30},
  number  = {1},
  pages   = {50--64},
  year    = {1951},
  month   = jan,
  doi     = {10.1002/j.1538-7305.1951.tb01366.x}
}

@article{Duchin_Tapp_2026, title={Metric geometry for ranking-based voting: Tools for learning electoral structure}, url={http://arxiv.org/abs/2602.10293}, DOI={10.48550/arXiv.2602.10293}, note={arXiv:2602.10293}, number={arXiv:2602.10293}, publisher={arXiv}, author={Duchin, Moon and Tapp, Kristopher}, year={2026}, month=feb }

@article{Szufa_Boehmer_Bredereck_Faliszewski_Niedermeier_Skowron_Slinko_Talmon_2025, title={Drawing a Map of Elections}, url={http://arxiv.org/abs/2504.03809}, DOI={10.48550/arXiv.2504.03809}, note={arXiv:2504.03809}, number={arXiv:2504.03809}, publisher={arXiv}, author={Szufa, Stanisław and Boehmer, Niclas and Bredereck, Robert and Faliszewski, Piotr and Niedermeier, Rolf and Skowron, Piotr and Slinko, Arkadii and Talmon, Nimrod}, year={2025}, month=apr }

@article{Faliszewski_Mertlova_Nunn_Szufa_Was_2026,
  author = {Faliszewski, Piotr and
            Mertlov{\'a}, Jitka and
            Nunn, Pierre and
            Szufa, Stanis{\l}aw and
            W{\k a}s, Tomasz},
  title = {Distances Between Top-Truncated Elections of Different Sizes},
  journal = {arXiv},
  year = {2026},
  month = jan,
  note = {arXiv:2601.17931},
  doi = {10.48550/arXiv.2601.17931},
  url = {http://arxiv.org/abs/2601.17931}
}

@article{Sawadogo_Dossou-Gbete_Lafon_2017,
  title   = {Ties in one block comparison experiments: a generalization of the {Mallows--Bradley--Terry} ranking model},
  volume  = {44},
  number  = {14},
  journal = {Journal of Applied Statistics},
  author  = {Sawadogo, Amadou and Dossou-Gb{\'e}t{\'e}, Simplice and Lafon, Dominique},
  year    = {2017},
  month   = oct,
  pages   = {2621--2644},
  doi     = {10.1080/02664763.2016.1259400}
}

@article{Davidson_1970,
  author  = {Davidson, R. R.},
  title   = {On extending the {Bradley--Terry} model to accommodate ties in paired comparison experiments},
  journal = {Journal of the American Statistical Association},
  volume  = {65},
  year    = {1970},
  pages   = {317--328}
}

@book{Marden_1995,
  author    = {Marden, J. I.},
  title     = {Analyzing and Modeling Rank Data},
  publisher = {Chapman \& Hall},
  address   = {London},
  year      = {1995}
}

@article{Fligner_Verducci_1986,
  author  = {Fligner, M. A. and Verducci, J. S.},
  title   = {Distance based ranking models},
  journal = {Journal of the Royal Statistical Society, Series B (Methodological)},
  volume  = {48},
  number  = {3},
  pages   = {359--369},
  year    = {1986}
}

@article{Fligner_Verducci_1988,
  author  = {Fligner, M. A. and Verducci, J. S.},
  title   = {Multistage ranking models},
  journal = {Journal of the American Statistical Association},
  volume  = {83},
  number  = {403},
  pages   = {892--901},
  year    = {1988}
}

@article{Critchlow_Fligner_Verducci_1991,
  author  = {Critchlow, Douglas E. and Fligner, Michael A. and Verducci, Joseph S.},
  title   = {Probability Models on Rankings},
  journal = {Journal of Mathematical Psychology},
  volume  = {35},
  year    = {1991},
  pages   = {294--318}
}

@article{babingtonsmith1950discussion,
  author  = {Babington Smith, B.},
  title   = {Discussion of {P}rofessor {R}oss's paper},
  journal = {Journal of the Royal Statistical Society, Series B},
  year    = {1950},
  volume  = {12},
  number  = {1},
  pages   = {53--56}
}

@book{Luce_1959,
  author    = {Luce, R. Duncan},
  title     = {Individual Choice Behavior: A Theoretical Analysis},
  publisher = {John Wiley \& Sons},
  address   = {New York},
  year      = {1959}
}

@article{Plackett_1975,
  author  = {Plackett, R. L.},
  title   = {The Analysis of Permutations},
  journal = {Applied Statistics},
  volume  = {24},
  number  = {2},
  year    = {1975},
  pages   = {193--202}
}

@article{Bradley_Terry_1952,
  author  = {Bradley, Ralph Allan and Terry, Milton E.},
  title   = {Rank Analysis of Incomplete Block Designs: {I}. The Method of Paired Comparisons},
  journal = {Biometrika},
  volume  = {39},
  number  = {3/4},
  year    = {1952},
  pages   = {324--345}
}

@article{Mallows_1957,
  author  = {Mallows, C. L.},
  title   = {Non-Null Ranking Models. {I}},
  journal = {Biometrika},
  volume  = {44},
  number  = {1/2},
  year    = {1957},
  pages   = {114--130}
}

@article{Lu_Boutilier_2014,
  author  = {Lu, Tyler and Boutilier, Craig},
  title   = {Effective Sampling and Learning for {Mallows} Models with Pairwise-Preference Data},
  journal = {Journal of Machine Learning Research},
  volume  = {15},
  number  = {117},
  year    = {2014},
  pages   = {3963--4009}
}

@inproceedings{Busse_Orbanz_Buhmann_2007,
  author    = {Busse, Ludwig M. and Orbanz, Peter and Buhmann, Joachim M.},
  title     = {Cluster Analysis of Heterogeneous Rank Data},
  booktitle = {Proceedings of the 24th International Conference on Machine Learning},
  year      = {2007},
  pages     = {113--120},
  doi       = {10.1145/1273496.1273511}
}

@incollection{Beckett_1993,
  author    = {Beckett, Laurel A.},
  title     = {Maximum Likelihood Estimation in {Mallows's} Model Using Partially Ranked Data},
  editor    = {Fligner, Michael A. and Verducci, Joseph S.},
  booktitle = {Probability Models and Statistical Analyses for Ranking Data},
  publisher = {Springer-Verlag},
  address   = {New York},
  year      = {1993},
  pages     = {92--107},
  doi       = {10.1007/978-1-4612-2738-0_6}
}

@inproceedings{Liu_et_al_2019,
  author    = {Liu, Ao and Zhao, Zhibing and Liao, Chao and Lu, Pinyan and Xia, Lirong},
  title     = {Learning {Plackett--Luce} Mixtures from Partial Preferences},
  booktitle = {Proceedings of the AAAI Conference on Artificial Intelligence},
  volume    = {33},
  number    = {1},
  year      = {2019},
  pages     = {4328--4335},
  doi       = {10.1609/aaai.v33i01.33014328}
}

@inproceedings{Hajek_Oh_Xu_2014,
  author    = {Hajek, Bruce and Oh, Sewoong and Xu, Jiaming},
  title     = {Minimax-Optimal Inference from Partial Rankings},
  booktitle = {Advances in Neural Information Processing Systems 27},
  year      = {2014},
  pages     = {1475--1483}
}

@article{Doignon_Pekec_Regenwetter_2004,
  author  = {Doignon, Jean-Paul and Peke{\v{c}}, Aleksandar and Regenwetter, Michel},
  title   = {The Repeated Insertion Model for Rankings: Missing Link between Two Subset Choice Models},
  journal = {Psychometrika},
  volume  = {69},
  number  = {1},
  year    = {2004},
  pages   = {33--54},
  doi     = {10.1007/BF02295838}
}

@misc{Duchin_Tapp_MaxEnt,
  author={Duchin, Moon and Tapp, Kristopher},
  title  = {Maximum entropy methods for election generation},
  note   = {In preparation},
  year   = {2026}
}

@article{jaynes1957information,
  author  = {Jaynes, Edwin T.},
  title   = {Information Theory and Statistical Mechanics},
  journal = {Physical Review},
  volume  = {106},
  number  = {4},
  pages   = {620--630},
  year    = {1957},
  doi     = {10.1103/PhysRev.106.620}
}

@inproceedings{awadelkarim2024topk,
  author    = {Awadelkarim, Amel and Ugander, Johan},
  title     = {Statistical Models of Top-$k$ Partial Orders},
  booktitle = {Proceedings of the 30th ACM SIGKDD Conference on Knowledge Discovery and Data Mining (KDD '24)},
  year      = {2024},
  pages     = {39--48},
  publisher = {ACM},
  doi       = {10.1145/3637528.3672014}
}

@article{hoffman2025truncation,
  author  = {Hoffman, Christina and Kauba, Jakini and Reidy, Julie C. and Weighill, Thomas},
  title   = {Statistical models of ballot truncation in ranked choice elections},
  journal = {Communications in Statistics -- Simulation and Computation},
  volume  = {54},
  number  = {11},
  pages   = {4908--4926},
  year    = {2025},
  doi     = {10.1080/03610918.2024.2397032}
}

@article{diaconis1989spectral,
  author  = {Diaconis, Persi},
  title   = {A Generalization of Spectral Analysis with Application to Ranked Data},
  journal = {The Annals of Statistics},
  volume  = {17},
  number  = {3},
  pages   = {949--979},
  year    = {1989},
  doi     = {10.1214/aos/1176347251}
}

@inproceedings{nguyen2023plackettluce,
  author    = {Nguyen, Duc and Zhang, Anderson Y.},
  title     = {Efficient and Accurate Learning of Mixtures of {P}lackett--{L}uce Models},
  booktitle = {Proceedings of the AAAI Conference on Artificial Intelligence},
  volume    = {37},
  number    = {8},
  pages     = {9294--9301},
  year      = {2023},
  doi       = {10.1609/aaai.v37i8.26114}
}

\clearpage
\appendix
\section{Maximum entropy}\label{sec:maxentropy}
Let $P$ be a profile for which every candidate is mentioned on at least one ballot.  
Let $M=(m_{ij})$ be its succession matrix, and let $Y=(y_{ij})$ be defined as in the proof of Proposition~\ref{prop:steady_state}; that is, $y_{ij}$ is the absolute frequency of the substring $ij$.  We call $Y$ the \emph{candidate-pairs} matrix of $P$.  

In this section, we describe the \MaxEntropy method, which samples ballots from the distribution $\nu_0$ that has maximum Shannon entropy subject to the constraint that its succession matrix equals $M$.  More details can be found in a sequel paper that applies maximum entropy methods to a variety of election profile summaries~\cite{Duchin_Tapp_MaxEnt}.

Since \MaxEntropy will turn out to have a similar formula, we first note that the unnormalized probability that a ballot $\sigma = (A_{i_1}, A_{i_2},\ldots,A_{i_l})$ produced by Autocomplete-\REJ is
\begin{equation}
\P(\sigma) \propto m_{\star i_1}m_{i_1i_2}m_{i_2i_3}\cdots m_{i_l\star}.
\label{E:rej_prob}
\end{equation}
The \REJ method is not faithful to succession probabilities because these probabilities change when conditioned on a raw ballot being valid. The \MaxEntropy probabilities will have the same form, but with the values $m_{ij}$ in Equation~\ref{E:rej_prob} replaced by the entries of a matrix $W=(w_{ij})$ that is tuned so that the resulting probability distribution is faithful to the original succession matrix $M$. However, since $W$ is not necessarily row-stochastic, sampling cannot in general be performed via rejection; instead, we build the resulting distribution directly.

With notation as in Section~\ref{sec:iteration}, For $\nu\in\Delta(\B)$, we will write $M(\nu)$ for the succession matrix of $\nu$, which is denoted $g(\nu)$ in Section~\ref{sec:iteration}.  We additionally define the \emph{candidate pairs} matrix $Y(\nu)$ of $\nu$ so that $Y(\nu)_{ij}$ is the absolute frequency of the sub-string $ij$ in a ballot drawn from $\nu$. 

Since $M$ is the succession matrix of an actual profile, the set 
$$S= \{\nu\in\Delta(\B)\mid M(\nu)=M\}$$ is nonempty and compact.  Our goal is to maximize Shannon entropy in $S$. For this, it is convenient to first re-describe $S$ as 
$$S= \{\nu\in\Delta(\B)\mid Y(\nu)=Y\},$$
This characterization is valid because, according to Proposition~\ref{prop:steady_state}, $Y$ can be determined from $M$.  Conversely, $Y$ also determines $M$, since
$m_{ij}=\frac{y_{ij}}{\mu_i} = \frac{y_{ij}}{\sum_j y_{ij}},$ where $\mu$ is the mentions vector.  Thus, $M(\nu)=M \iff Y(\nu)=Y$.  

This alternative characterization of $S$ makes it easier to see that $S$ is defined by linear equations and is therefore convex.  Since entropy is continuous and strictly concave, it has a unique maximizer on $S$.  Let $\nu_0$ be this maximizer -- the maximum entropy distribution for which the candidate-pairs matrix equals $Y$.  According to~\cite{jaynes1957information}, this distribution has the following simple form.

\begin{proposition}\label{P:lagrange} If $\nu_0$ has full support, then there exists an $(m+1)$-by-$(m+1)$ matrix $W=(w_{ij})$, depending on $Y$, such that for every ballot $\sigma = (A_{i_1}, A_{i_2},...,A_{i_l})$,
$$\P_{\nu_0}(\sigma) \propto w_{\star i_1}\cdot w_{i_1i_2}\cdot w_{i_2i_3} \cdots w_{i_l \star}.$$   
\end{proposition}

To implement the \MaxEntropy method, given the candidate-pairs matrix $Y=Y(P)$, we determine the matrix $W = (w_{ij})$ for which the distribution described by Proposition~\ref{P:lagrange} has candidate-pairs matrix $Y$.  Equivalently, writing $w_{ij}=e^{\alpha_{ij}}$, we choose the parameters $\alpha_{ij}$ to satisfy this constraint.  This is the standard optimization problem associated with maximum entropy under linear constraints.  Its objective is convex, and its gradient records the discrepancies between the desired constraints and their expectations under the current distribution.  We solve this optimization numerically, obtaining a fitted matrix $W$, and then sample ballots from the resulting distribution.  More details are found in~\cite{Duchin_Tapp_MaxEnt}.

In the generic case where the maximum-entropy distribution $\nu_0$ has full support, this implementation of \MaxEntropy draws from $\nu_0$, as advertised.  In the general case, the algorithm still deserves its name.  When $\nu_0$ does not have full support, it is a limit of distributions of this form, with some of the parameters $\alpha$ tending to infinity in magnitude.  Thus our exponential family of distributions includes $\nu_0$ either directly or as a boundary limit. 


\section{Details on ballot statistics}\label{sec:ballotstats}

\begin{proposition}[Reconstitution shrinks ballots]\label{P:ballot_shrinking}
    If $P$ is not ordered, then the expected average ballot length in $\tP$ under Autocomplete with \Terminate, \KeepFirst, or \LERW is strictly less than the average ballot length in $P$.
\end{proposition}
\begin{proof}
    If $P$ is not ordered, then there is a positive probability that a raw ballot will have a repeated symbol.  Recall that the expected raw ballot length equals the average ballot length in $P$, so the expected processed ballot length will be strictly shorter.
\end{proof}

With the \REJ and \SARW generation methods, length shrinkage holds for typical profiles, but not for all profiles.
\begin{example}
    Consider the profile $P$ containing equal portions of the ballots
    $$\{(A_1), (A_2), (A_1,A_2), (A_2,A_1), (A_3,A_4,\cdots, A_{k+3})\}$$ for large $k$, and generate new ballots by Autocomplete-\REJ.  The expected length of a \emph{rejected} raw ballot is smaller than that of an \emph{accepted} raw ballot because the rejected ones contain only the candidates $A_1$ and $A_2$.  It follows that the expected average ballot length in $\tP$ is larger than the average ballot length in $P$.
\end{example}

\begin{example}
Consider the profile consisting of $19$ copies of the ballot $(A_3)$, three
copies of $(A_3,A_4,A_1,A_2)$, and one copy of $(A_4,A_3,A_2)$. The average
ballot length of the original profile is $\frac{34}{23}\approx 1.47826$,
whereas the expected ballot length under Autocomplete-\SARW is
$\frac{63163}{42320}\approx 1.49251$.  The intuition is that walks beginning with $\star 3 4$ cannot return to $3$ and are thus forced onto the longer path $...12...$, which slightly more than compensates for the increased probabilities of stopping.
\end{example}

Recall that $T_i$ is the stopping weight of candidate $i$, which is the probability that they are placed last (when ranked at all).  
Let $\tilde{T}_i$ denote the stopping weight of $f_R(M)$; that is, the conditional probability that a ballot generated by $M$ via method $R$ places candidate $i$ last, given that it ranks $i$.

\begin{proposition}\label{P:increase_stop_weights}
(Reconstitution increases stopping weights.)
For each of the five random-walk methods
\REJ, \Terminate, \KeepFirst, \LERW, and \SARW,
$\widetilde T_i \geq T_i$
for every candidate $i$.
Moreover, if the profile $P$ is full-circulation (or more generally if its succession matrix is 2-sided and candidate $i$ belongs to a positive-weight candidate cycle), then the inequality is strict: $\widetilde T_i > T_i$.
\end{proposition}

\begin{proof}
Fix a candidate $i$.

For \Terminate, condition on a duplication-free walk prefix ending at $i$.
The ballot ends at $i$ either if the next step is $\star$, which has
probability $T_i$, or if the next step is a previously visited candidate.
Thus the conditional probability of ending at $i$ is at least $T_i$,
and hence $\widetilde T_i\geq T_i$.

For \SARW, again condition on a self-avoiding prefix ending at $i$.
The weights of edges from $i$ to previously visited candidates are
discarded and the remaining weights are renormalized.  In particular,
the stopping weight $T_i$ can only increase, so
$\widetilde T_i\geq T_i$.

For \KeepFirst, condition on the first visit to $i$.  The event that the
next step is $\star$ has probability $T_i$ and certainly causes the
generated ballot to end at $i$.  There may additionally be paths that
visit only previously seen candidates before reaching $\star$, and these
also produce ballots ending at $i$.  Hence
$\widetilde T_i\geq T_i$.

For \REJ, condition on a duplication-free prefix ending at $i$, and
then condition further on the eventual raw ballot being valid.  If the
next step is $\star$, acceptance is certain, whereas a transition to
another candidate has acceptance probability at most $1$.  Conditioning
on acceptance therefore weakly favors the transition $i\to\star$, so
again $\widetilde T_i\geq T_i$.

Finally, consider \LERW.  Let $N_i$ be the number of visits to $i$ in a
raw ballot.  By Proposition~\ref{prop:steady_state}, $\mathbb E[N_i]=\mu_i$.
Loop erasure preserves the endpoint of the raw ballot, so
\[
    \Pr(\text{\LERW ballot ends at }i)=\mu_iT_i.
\]
If $A_i$ is the event that $i$ survives loop erasure, then
\[
    \widetilde T_i
      = \frac{\mu_iT_i}{\Pr(A_i)}.
\]
Since $A_i$ implies that the raw walk visits $i$,
\[
    \Pr(A_i)\leq \Pr(N_i\geq1)\leq\mathbb E[N_i]=\mu_i,
\]
and therefore $\widetilde T_i\geq T_i$.

It remains to prove strictness under the hypothesis that the profile is two-sided and $i$ is on a cycle.
Fix a positive candidate cycle through $i$.

For \REJ, start the walk at $i$ and traverse the cycle back to $i$.
This gives a positive-probability continuation that is rejected, while
$i\to\star$ is accepted, so conditioning on acceptance strictly
increases the stopping probability at $i$.

For \LERW, the same walk visits $i$ at least twice with positive
probability.  Hence
\[
    \Pr(A_i)<\mathbb E[N_i]=\mu_i,
\]
and, since $T_i>0$, the displayed formula above gives
$\widetilde T_i>T_i$.

For each of \Terminate, \SARW, and \KeepFirst, choose a substring
$ij$ on the cycle and begin instead at $j$, then follow the remainder
of the cycle until reaching $i$.  This occurs with positive probability
by the circulation hypothesis, and when $i$ is reached, $j$ has already been
visited.  Thus the occurrence of $ij$ gives additional stopping probability
under \Terminate and is removed and renormalized away under \SARW,
making the inequality strict.  Under \KeepFirst, after seeing
$ij$, the repeated $j$ is not recorded; following with
$j\star$ (which again has positive probability by the circulation hypothesis), produces a ballot ending at $i$ without using the transition $i\star$.
Thus the inequality is strict for \KeepFirst as well.
\end{proof}

\section{Details on empirical comparisons}
\label{sec:empirical-details}

Figure~\ref{fig:dist-to-P} (top) demonstrates that for the purposes of reproducing $P$ itself, succession methods perform strongly.  All the succession methods have a median performance that is stronger than position- or H2H-based regeneration, with \MaxEntropy attaining particularly close matches.

Figure~\ref{fig:dist-to-P} (bottom) is similar but replaces the earth-mover distance with the \emph{positionwise distance} between $P$ and $\tP$, as defined in~\cite{Szufa_Boehmer_Bredereck_Faliszewski_Niedermeier_Skowron_Slinko_Talmon_2025} for profiles of complete ballots and generalized in~\cite{Faliszewski_Mertlova_Nunn_Szufa_Was_2026} to general profiles.  This is defined by first constructing the \emph{frequency matrices} of $P$ and $\tP$, which means the $m$-by-$m$ matrix whose $(i,j)$-entry reports the probability the $j^\text{th}$ candidate is ranked in the $i^\text{th}$ position by a random completion of a random ballot from the profile.  The positionwise distance is then defined as the sum over candidates $i$ of the Wasserstein distance between the columns of the frequency matrices of $P$ and $\tP$ corresponding to that candidate.

The \MaxEntropy method outperforms the positional method even when the measure of differences is based on positional frequencies.  

\begin{figure}[htb!]
    \centering
    \includegraphics[width=4.4in]{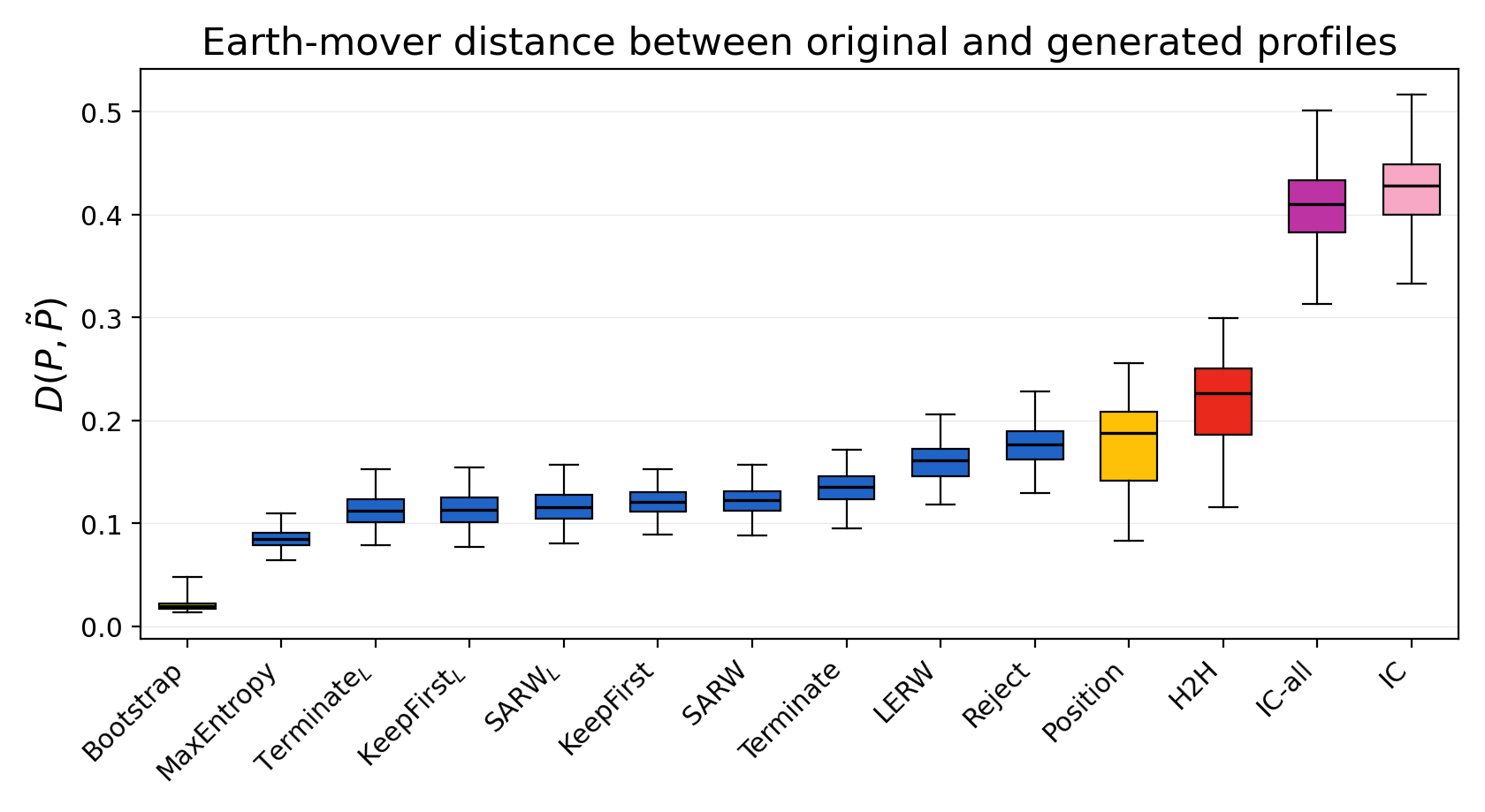}
    \includegraphics[width=4.4in]{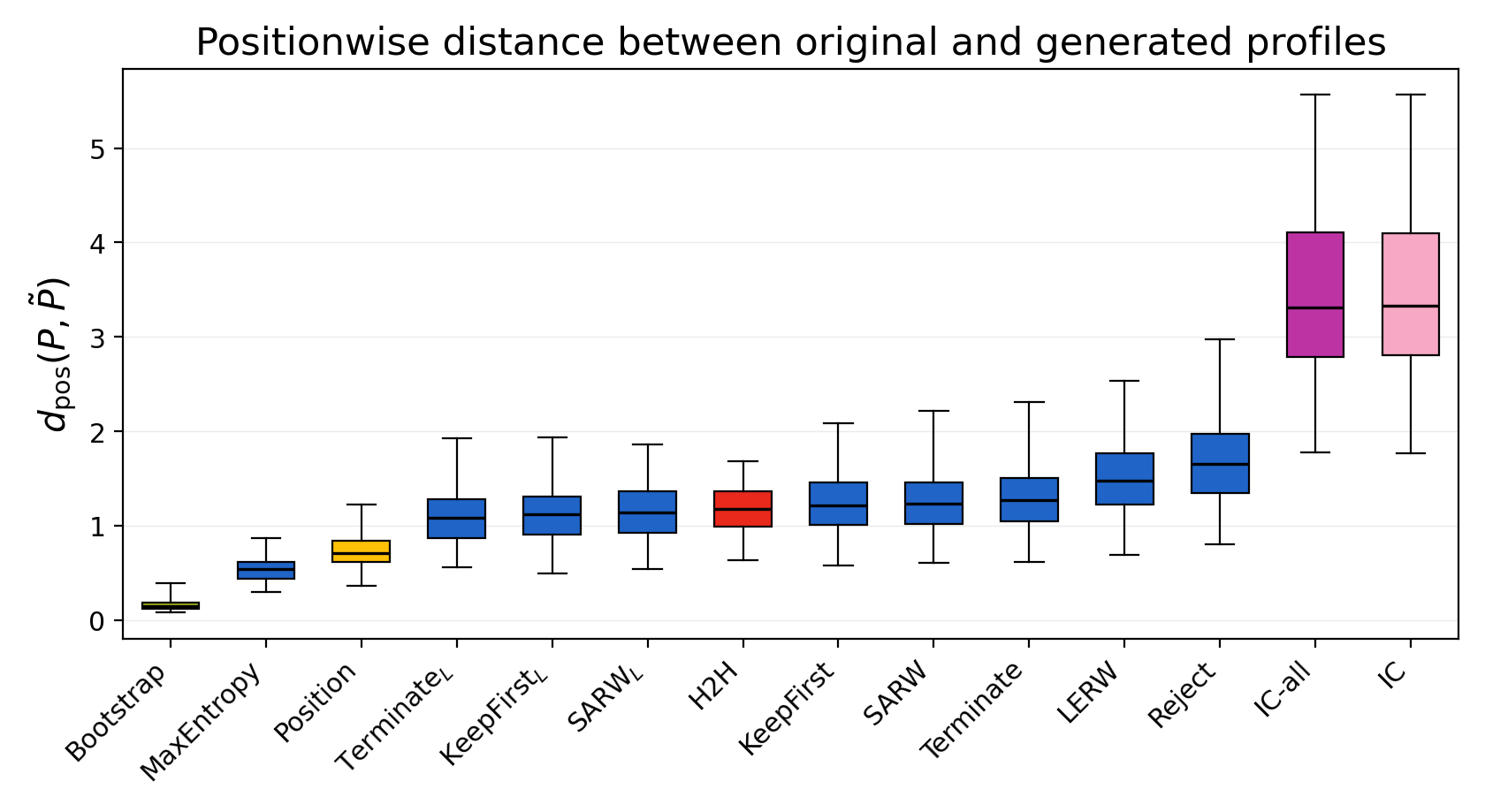}
    \caption{The distribution of $D(P,\tP)$ (top) and positional distance (bottom) over the 282 seven-candidate Scottish elections. In each plot, the models are ordered by increasing medians.  Boxes show the interquartile range, the center line is the median, and whiskers show the $1^\text{st}$ and $99^\text{th}$ percentiles.}
    \label{fig:dist-to-P}
\end{figure}

\FloatBarrier

Figure~\ref{fig:STV-k3} verifies that the 3-winner results are very similar to 4-winner results previously shown in Figure~\ref{fig:STV-k4}.
\begin{figure}[htb!]
    \centering
    \includegraphics[width=4.2in]{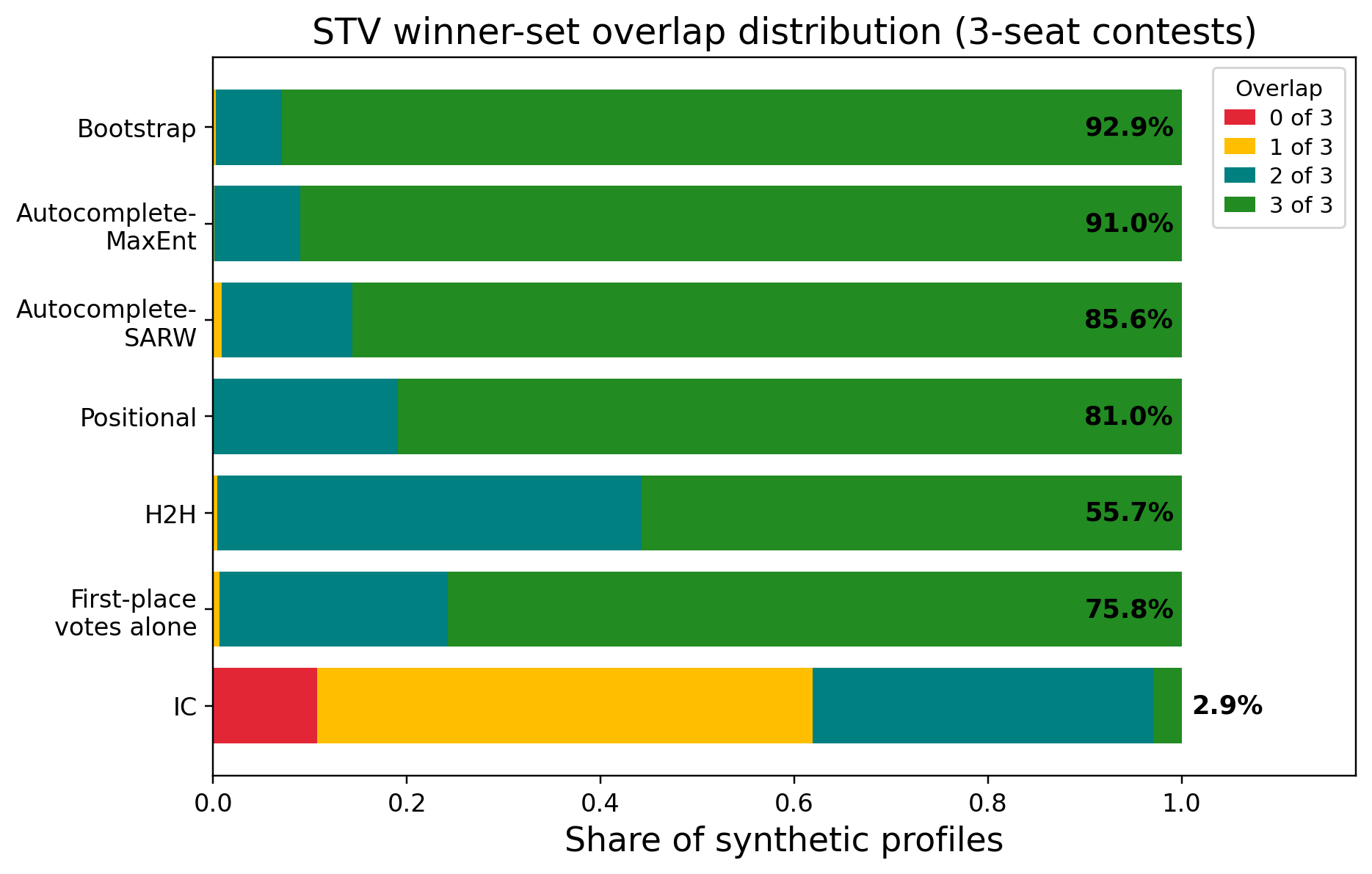}
    \caption{If the reconstituted preferences are fed to an STV voting rule, how much does the winner-set overlap with the original $f(P)$?  This figure shows the results for the 157 Scottish contests with 7 candidates and 3 winners.}
    \label{fig:STV-k3}
\end{figure}

\section{Details on iteration dynamics}\label{sec:iteration-details}

This appendix gives the setup, statements, and proofs behind Theorem~\ref{thm:iteration} and Proposition~\ref{prop:twocand}, followed by the computed limits for our highlighted elections.

\subsection{Setup}\label{sec:iteration_dynamics}

As in Section~\ref{sec:iteration}, for a given ballot generation rule $R$, we have maps as follows:
$\begin{tikzcd}
\Delta(\B)
  \arrow[r, bend left=20, "g"]
&
\M.
  \arrow[l, bend left=20, "f_R"]
\end{tikzcd}$
We will write $f$ when the rule $R$ is fixed, and $\Phi_R = g\circ f_R:\M\to\M$ for the induced map on matrices.

Except in Appendix~\ref{sec:iteration-reject}, we assume throughout this section that $$R\in\{\Terminate,\SARW, \LERW, \KeepFirst\}.$$  Most results are also true of \REJ, as discussed in Appendix~\ref{sec:iteration-reject}.  \MaxEntropy is excluded because its dynamics are trivial: it is faithful to the succession matrix by construction, so $\Phi_{\MaxEntropy}$ is the identity.  The interesting dynamics are generated precisely by the lossiness of the cheaper reconstitution rules.

\begin{remark}
$\Delta(\B)$ has $O(m!)$ parameters, while $\M$ has  $O(m^2)$, so many measures are not achieved by transition matrices.  An example  is any measure whose support contains only the ballots $(A_1,A_2,A_3)$ and $(A_2)$.  For any matrix $M$ to be compatible with this measure, all of the weights $m_{\star 1}$, $m_{12}$, and $m_{2\star}$ will need to be positive, but this means that some $(A_1,A_2)$ ballots will occur, no matter the choice of processing rule $R$.
\end{remark}

\begin{definition}\label{def:fixed-twosided}
A probability distribution $\nu\in\Delta(\B)$ is fixed under iteration if $f\circ g(\nu)=\nu$.  Likewise, a transition matrix is fixed under iteration if $\Phi_R(M)=M$.

Recall from \S\ref{sec:prelims} that $M$ is {\em acyclic} if it has no candidate cycle.  A distribution $\nu\in\Delta(\B)$ is {\em ordered} if there exists a (possibly not unique) total order of the candidates so that every ballot in the support of $\nu$ respects that order.  Recall from \S\ref{sec:iteration} that $M$ is {\em two-sided} if $m_{\star i}>0$ and $m_{i \star}=T_i>0$ for every candidate $i$.  Two-sidedness ensures that $f_R(M)$ ranks every candidate with positive probability, so that $\Phi_R(M)$ is defined; the succession matrix of a full-circulation profile is two-sided. Two-sidedness is also fixed under iteration, as we will see below, while the full circulation property is not.
\end{definition}

\begin{remark}\label{R:closure} Let $M\in\M$, $\nu\in\Delta(\B)$, $R$ be one of the four processing rules under consideration.
\begin{enumerate}
    \item If $M$ is acyclic, then all processing rules are equivalent.
    \item $M$ is acyclic iff $f_R(M)$ is ordered, and $\nu$ is ordered iff $g(\nu)$ is acyclic.
    \item If $M$ is acyclic, then $\Phi_R(M)$ is also acyclic.
    \item If $\nu$ is ordered with respect to a permutation $\tau$, then $f_R\circ g(\nu)$ is also ordered with respect to $\tau$.

\end{enumerate}
\end{remark}

\subsection{Fixed points and convergence}

\begin{proposition}[Invariance and monotonicity]\label{P:invariants}
If $M$ is two-sided and $M'=\Phi_R(M)$, then:
\begin{enumerate}
    \item $M'$ is two-sided, with the same $\star$ row as $M$;
    \item $T_i(M')\geq T_i(M)$ for every candidate $i$; and
    \item $T_i(M')>T_i(M)$ for every candidate $i$ that lies on a positive-weight candidate cycle in $M$.
\end{enumerate}
\end{proposition}

\begin{proof}
Each of the four rules retains the first symbol of the raw walk, so $m_{\star i}' = m_{\star i}$ for every $i$.
Since $M$ is two-sided, the bullet vote for any candidate has positive probability, and hence $M'$ is two-sided.  The remaining claims follow from Proposition~\ref{P:increase_stop_weights}.
\end{proof}

In particular, this means that all stopping weights strictly increase for a full-circulation profile.

\begin{proposition}[Fixed points]\label{prop:fixed}
If $M\in\M$ is two-sided, then $M$ is fixed under iteration if and only if $M$ is acyclic.
\end{proposition}

\begin{proof}
If $M$ is not acyclic, then $\Phi_R(M)\neq M$ by Proposition~\ref{P:invariants}(3).
Conversely, assume that $M$ is acyclic.  The probability of a given ballot $\sigma$ in $\nu=f_R(M)$ is just the product of its sequence of weights (regardless of the rule $R$), because no loops means no reductions.

This implies that, for any pair of candidates $i,j$, the conditional probability in $\nu$ that candidate $j$ immediately follows $i$ on a ballot, given that the ballot contains $i$, is equal to $m_{ij}$.  This is exactly the weight of the edge from $i$ to $j$ in $\Phi_R(M)$.
\end{proof}

An ordered distribution, by contrast, might need one step of iteration before it becomes fixed.
That is, if $\nu$ is ordered, then $g(\nu)$ is acyclic, so $f\circ g(\nu)$ is fixed under iteration. The following example shows that this one extra step may be needed.

\begin{example}\label{Ex:delay}
Consider a distribution $\nu$ with half of the probability weight on the bullet vote $(A_3)$ and half on the complete ballot $(A_1,A_2,A_3,A_4,A_5)$.  Though this is ordered, it is not fixed, because $f\circ g(\nu)$ is equally supported on $(A_3), (A_1,A_2,A_3), (A_3,A_4,A_5)$, and $(A_1,A_2,A_3,A_4,A_5)$.
\end{example}

Write $M_n = \Phi_R^n(M_0)$ for the orbit points, starting from a two-sided matrix $M_0$.  Since the stopping weights can only increase along the orbit and are bounded, they converge.

\begin{proposition}[Limit points are acyclic]\label{prop:limits-acyclic}
If $M_0$ is two-sided, then every limit point of the orbit $(M_n)$ is an acyclic matrix fixed by $\Phi_R$.  Moreover, if $R=\Terminate$, then the orbit has a unique limit.
\end{proposition}

\begin{proof}
We first claim that the map $\Phi_R=g\circ f_R$ is continuous on the region
$$\tilde\M = \{M\in\M\mid M \text{ is two-sided, and } T_i \geq T_i(M_0) \text{ for all } i\}.$$   The map $g:\Delta(\B)\rightarrow \M$ is continuous where it is defined, which is on the set of distributions for which each candidate has a positive probability of occurring on a ballot.  In particular, $g$ is defined and continuous on $f_R(\tilde\M)$ because of the two-sided hypothesis, so it suffices to prove that $f_R$ is continuous.  For $R\in\{\Terminate, \SARW\}$, the probability of any given ballot $\sigma$ under $f_R(M)$ is a rational function of the entries of $M$.  For $R\in\{\LERW, \KeepFirst\}$,
$$\P(\text{raw ballot's length}>k)\leq(1-\epsilon)^k,$$
where $\epsilon=\min\{T_i(M_0)\}>0$.   For any fixed $k$, the distribution obtained by restricting to raw ballots of length at most $k$ is continuous in $M$.  The bound above is uniform over $\tilde{\mathcal M}$, so these finite approximations converge uniformly to $f_R$, proving continuity.

Any limit point $M^*$ is two-sided: its $\star$ row equals that of $M_0$, and $T_i(M^*)\geq T_i(M_0)>0$ by monotonicity.  Thus, the region $\tilde\M$ contains the orbit and its limit points.

Let $S(M)=\sum_i T_i(M)$ (the sum of the stopping weights), which is non-decreasing along the orbit by Proposition~\ref{P:invariants}(2), hence convergent.  If a limit point $M^*$ had a positive-weight candidate cycle, then $S(\Phi_R(M^*))>S(M^*)$ by Proposition~\ref{P:invariants}(3); taking a subsequence $M_{n_k}\to M^*$ and using continuity, $S(M_{n_k+1})\to S(\Phi_R(M^*))$, which exceeds $\lim S(M_n)=S(M^*)$---a contradiction.  So every limit point is acyclic, and it is fixed by Proposition~\ref{prop:fixed}.

If $R=\Terminate$, then all entries of the succession matrix (not just the stopping weights) are monotonic; the stopping weights are non-decreasing, while it is straightforward to prove that the between-candidate weights are non-increasing.  It follows that the orbit converges to a unique limit.
\end{proof}

To quantify the approach to acyclicity, we use the $\alpha$-cyclicity of Definition~\ref{D:cyclicity}.

\begin{corollary}\label{C:cyclicity}
For any $0<\alpha<1$, the $\alpha$-cyclicity $\cyc(M_n,\alpha)$ converges to $0$.
\end{corollary}

\begin{proof}
    Suppose to the contrary that $\cyc(M_n,\alpha)\nrightarrow 0$, which means that a subsequence has $\alpha$-cyclicity $>\epsilon>0$.  By compactness, there is a convergent sub-subsequence, whose limit is acyclic according to Proposition~\ref{prop:limits-acyclic}.  Since $\cyc(\cdot ,\alpha)$ is continuous, this contradicts the fact that the $\alpha$-cyclicity of an acyclic matrix is zero.
\end{proof}

Appendix~\ref{sec:iteration-plots} illustrates the drop in cyclicity under iteration on the Scottish data.

\subsection{Acyclicity for \REJ limits}\label{sec:iteration-reject}

When $R=\REJ$, Remark~\ref{R:closure} and Proposition~\ref{P:invariants}(2)-(3) and Proposition~\ref{prop:fixed} remain true; however, Proposition~\ref{prop:limits-acyclic} breaks down for the following reason.  A candidate's first-place votes are not preserved by iteration, and might tend to zero, in which case the limit $M^*$ is not a fixed point of $\Phi_R$ because it is not even in the domain of $\Phi_R$.  Nevertheless, we show next that $M^*$ does have a property that is slightly weaker than acyclic.

\begin{proposition}\label{prop:reject-limits}
If $R=\REJ$, then any limit point of the orbit $(M_n)$ has the property that no candidate that can be reached from $\star$ by a positive-weight path can lie on a positive-weight candidate cycle.
\end{proposition}

\begin{proof}
Let $M^*$ be a limit point, and let $\A$ be the set of candidates who can be reached from $\star$ by a positive-weight path in the succession graph associated to $M^*$.  Let $\tilde{M}^*$ be the matrix obtained from $M^*$ by removing the row and column of each candidate not in $\A$, which is the transition matrix for a Markov chain $\tilde{\Omega}$ on the smaller state space $\A\cup\{\star\}$.

Suppose to the contrary that the proposition is false; that is, there exists a candidate $i\in\A$ who lies on a positive-weight candidate cycle of $\tilde{\Omega}$; in particular, $\tilde{M}^*$ is not acyclic.  As in the previous proof, $T_i(M_n)$ converges monotonically to $T_i(M^*)$.

The proof of Proposition~\ref{P:increase_stop_weights} in the case $R=\REJ$ did not require the full two-sided hypothesis -- this proof shows that $\Phi_R$ strictly increases the stopping weight of any candidate that is contained on a cycle of positive weight and that can be reached from $\star$ by a positive-weight path.  In particular, $T_i(\Phi_R(\tilde{M}^*))>T_i(\tilde{M}^*)$.

On the other hand, the effect on $T_i(M_n)$ of the candidates outside of $\A$ goes to zero.  Indeed, along the convergent subsequence, every transition from $\A$ to its complement tends to zero.  Moreover, the stopping weights are uniformly bounded below by
$\epsilon=\min\{T_j(M_0)\}>0,$ so the raw ballot lengths have a uniform geometric tail.  It follows that the probability that a raw ballot ever leaves $\A$ tends to zero.  Thus the Reject process on $M_n$ is asymptotically equivalent to that on $\tilde M^*$, forcing $T_i(M_n)$ to converge both to $T_i(\tilde M^*)$ and to $T_i(\Phi_R(\tilde M^*))$, a contradiction.
\end{proof}

\subsection{Exact limits for two candidates}\label{app:twocands}

\begin{proof}[Proof of Proposition~\ref{prop:twocand}]
Two-sidedness gives $\alpha,\beta>0$, and $m_{AB}=\frac{c}{1-b}=\frac{c}{\alpha+d}$, $m_{BA}=\frac{d}{\beta+c}$. A walk $\star AB$ yields the ballot $(A,B)$ under either rule, since \Terminate cuts off a return to $A$ and \SARW blocks it; thus one iteration replaces $(a,b,c,d)$ by $(\alpha-c',\beta-d',c',d')$, where
$$c'=\frac{\alpha c}{\alpha+d},\qquad d'=\frac{\beta d}{\beta+c}.$$
Thus $\alpha,\beta$ are preserved and $c,d$ are non-increasing, so they converge to a fixed point $(c^*,d^*)$ of this map, which satisfies $c^*d^*=0$.
Let $\psi=\alpha c-\beta d$. A direct computation gives
$$\psi'=\psi\left(1-\frac{cd}{(\alpha+d)(\beta+c)}\right),$$
so the sign of $\psi$ never changes; since $\varphi=\alpha c'-\beta d'$ is $\psi$ after one step, it has the same sign. Moreover $c-c'\ge cd/(\alpha+\beta)$, so $\sum_n c_nd_n\le(\alpha+\beta)c$, and the product of these factors converges to a positive number. Hence if $\varphi>0$ then $\alpha c^*-\beta d^*>0$, which forces $c^*>0$ and $d^*=0$, with $s=c^*\le c$; the case $\varphi<0$ is symmetric; and if $\varphi=0$ then $\alpha c^*=\beta d^*$, so $c^*=d^*=0$.
\end{proof}

\subsection{Examples of limiting orders}\label{app:3limits}

For Pentland Hills, North Ayrshire, and Renfrewshire, we iterate the map $\Phi_R$ exactly on ballot distributions, which is feasible for $m=7$.
The comparison $\varphi_{ij}$ used in \S\ref{sec:iteration} is motivated as follows. Intuitively, the first-place weights $\alpha_i$ are some of the only measures of candidate support that are invariant under iteration; $m_{ij}$ measures the frequency of the $i\to j$ succession but it can be eroded under iteration, and first-place frequency $\alpha_i$ is a key source of mentions that is stable.  Experiments suggest that the $\varphi$ with an $\alpha_i^2m_{ij}$ term handily outperforms raw $m_{ij}$ and also outperforms the simpler $\alpha_im_{ij}$ at predicting limiting edges on random matrices $M$.

\begin{figure}[htb!]
\centering
\begin{tikzpicture}
\node at (0,4.4){\includegraphics[width=3in]{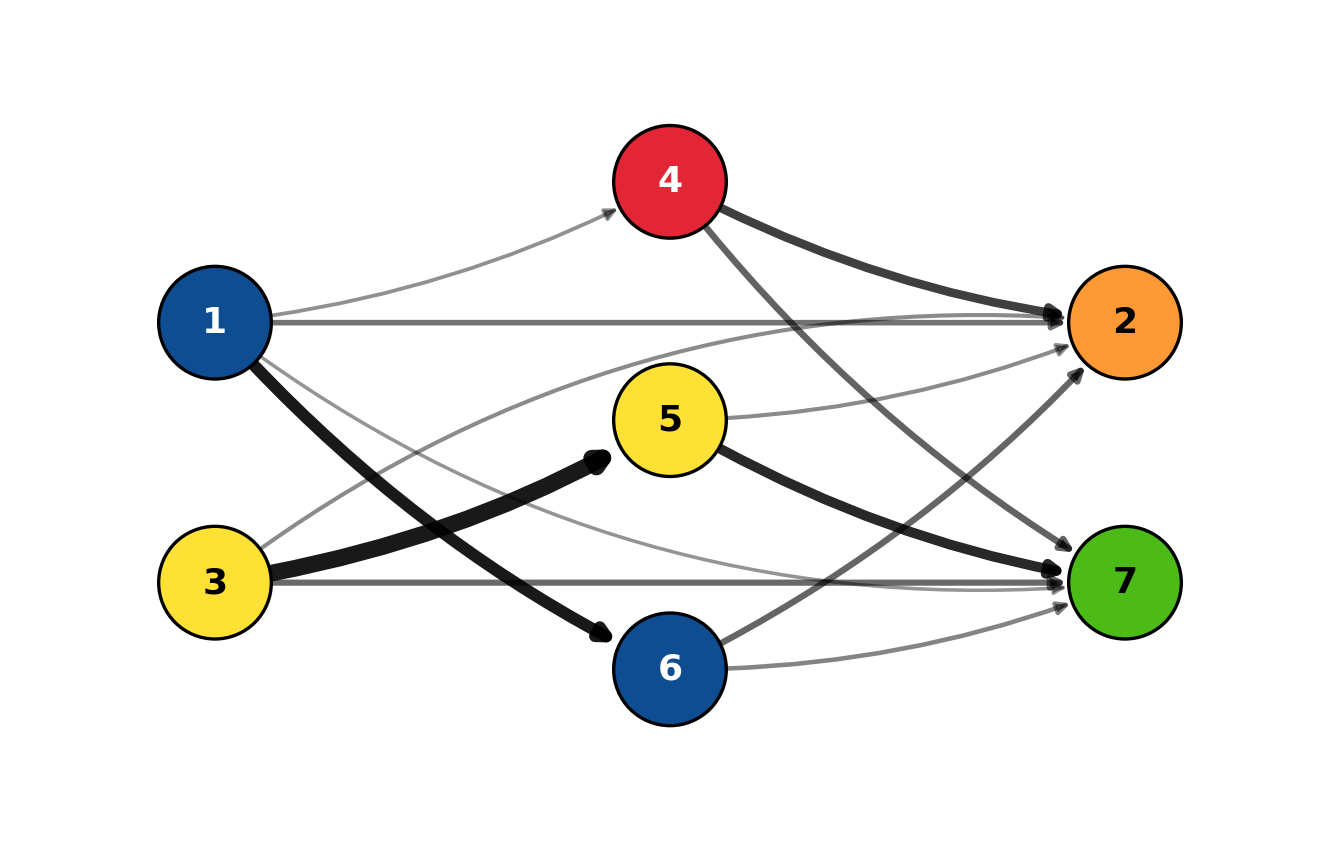}};
\node at (2.6,5.8){\SARW};

\node at (0,0){\includegraphics[width=3in]{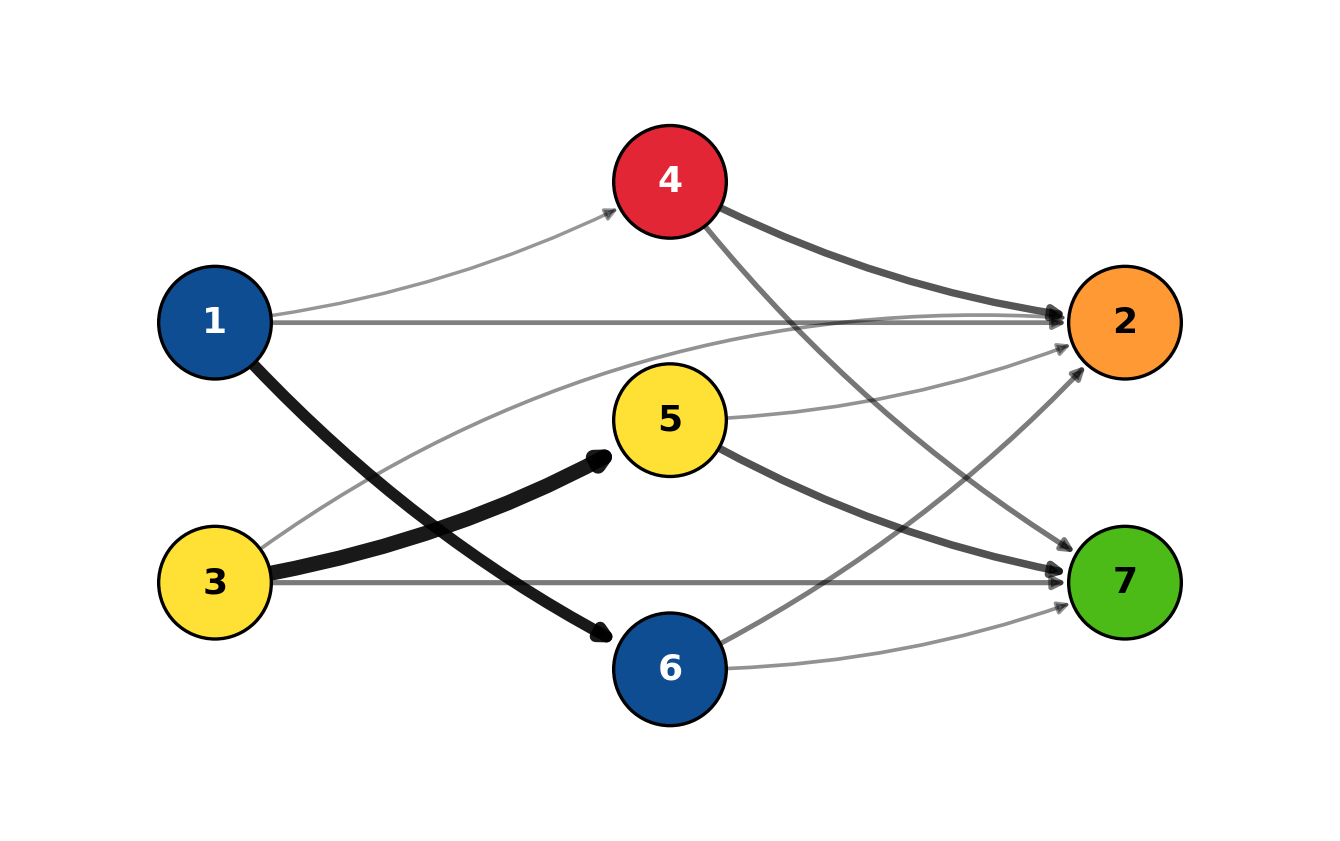}};
\node at (2.6,1.6){\Terminate};
\end{tikzpicture}

\caption{These are exact limits  for iterating the Pentland Hills profile under \SARW and \Terminate, respectively: each iteration uses the full ballot distribution rather than a sampled profile.
Arrows are shown for weights thresholded at $0.025$. The weights of the limiting order vary slightly, but $1\to 7$ is the only edge that is significant in one but not the other convention at this threshold.}\label{fig:limits}
\end{figure}

Figure~\ref{fig:limits} shows the edges of the $R=\Terminate,\SARW$ limits for Pentland Hills with a weight threshold applied for legibility.  At this threshold we see 12 (respectively, 13) edges under the two methods; the true limit in each case has 20 or 21 nonvanishing edges.\footnote{The \Terminate limiting order contains all $\binom 72=21$ edges consistent with the total order $1\succ3\succ4\succ5\succ6\succ2\succ7$, while the \SARW limit (which differs only in how it orders the middle tier) is consistent with $1\succ3\succ6\succ4\succ5\succ2\succ7$.  Of the 21 comparison edges, only the $6\to 5$ edge is ambiguous, carrying a limiting weight of around $8\times 10^{-10}$, near the edge of numerical precision.}   We use the standard colors of each political party in Scotland:  blue for Conservative, yellow for the (ruling, center-left) Scottish National Party, red for Labour, orange for Liberal Democrats, and green for Greens. This allows us to see that the strongest surviving successions follow the sequential order of candidates within parties.  (All ballots list the candidates in the same order  1,2,\dots, which is alphabetical by candidate last name, as prescribed by Scottish election law.)

\FloatBarrier
\section{Empirical convergence}\label{sec:iteration-plots}

Figures~\ref{F:loop1}-\ref{F:loop3} provide evidence consistent with fast convergence to unique acyclic limits with respect to all five methods.  Here the iteration is between a succession matrix and a specific random instance of a profile (rather than a distribution on ballots). 

\begin{figure}[htb!]
    \centering
\includegraphics[width=4in]{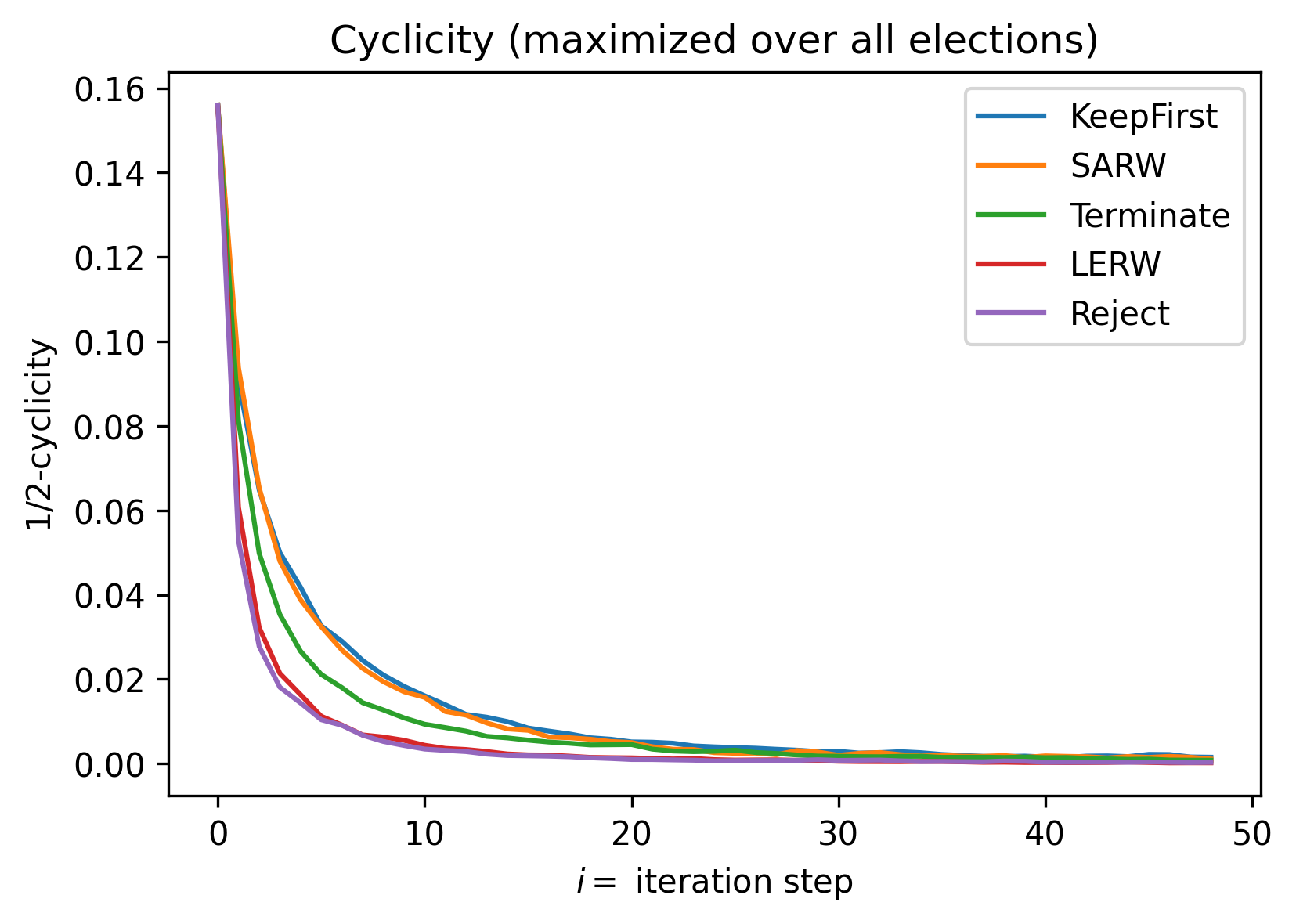}
\caption{The cyclicity plot shows the maximum over all seven-candidate elections of the $\frac{1}{2}$-cyclicity of $i^\text{th}$ iteration.}
    \label{F:loop1}\end{figure}

Figure~\ref{F:loop1} shows that the maximum (over all 282 seven-candidate Scottish elections) of $\cyc(M_n,1/2)$ seems to converge quickly to zero; thus, even in the worst elections, the matrices quickly approach acyclic.

\begin{figure}[htb!]
    \centering
\includegraphics[width=4in]{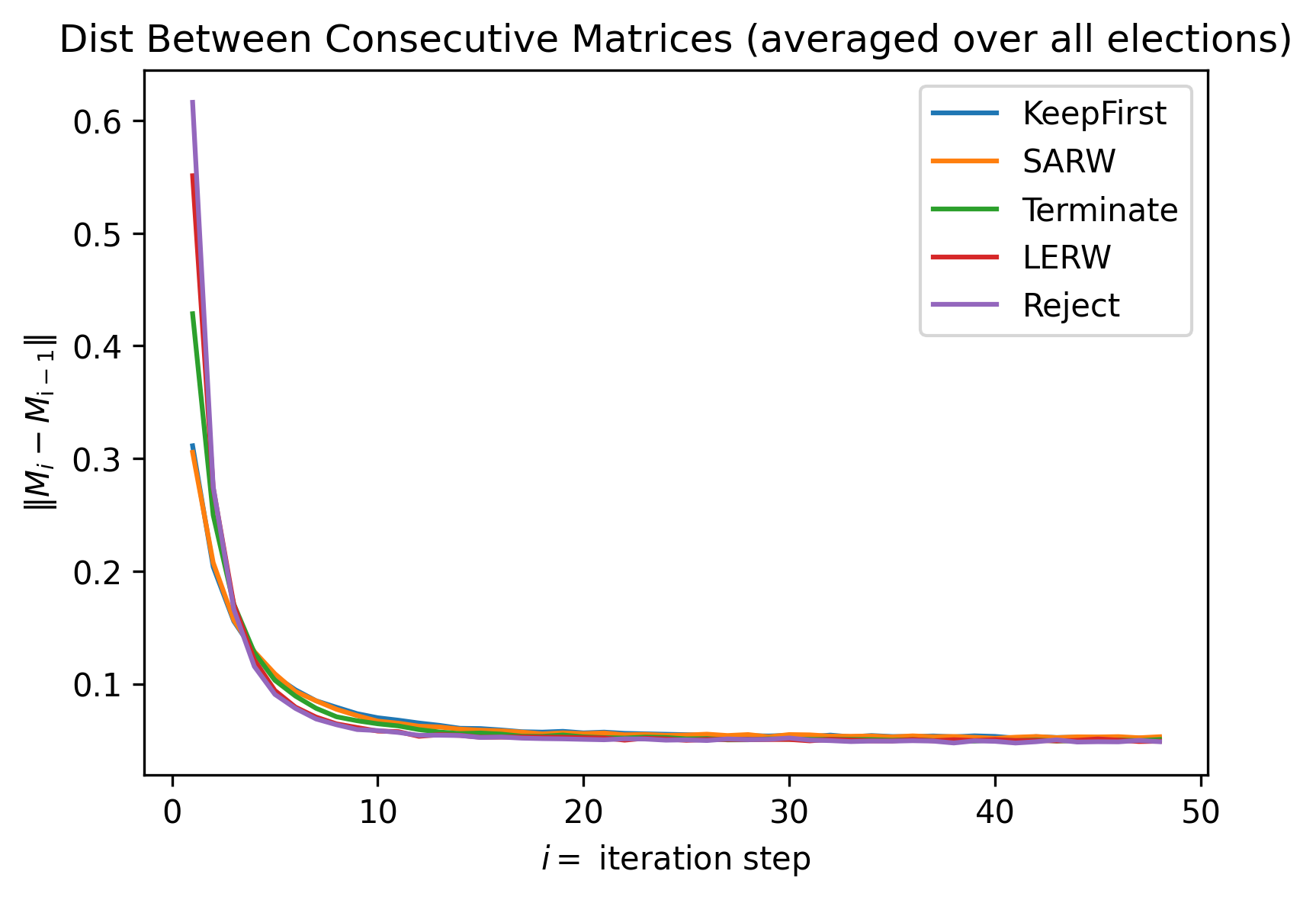}
    \caption{ The distance between consecutive matrices is shown as an average over all elections of $|M_i-M_{i-1}|$.}    \label{F:loop2}
\end{figure}

Figure~\ref{F:loop2}  provides evidence that consecutive iterates become arbitrarily close, consistent with convergence to a \emph{single} limit.  In this plot, we average (rather than maximize) over all elections to reduce noise caused by the randomness of reconstruction, since $M_i$ and $M_{i-1}$ can't be expected to be further apart from each other than two different instances of $M_i$ would be from each other.

\clearpage

\begin{figure}[htb!]
    \centering
\includegraphics[width=4in]{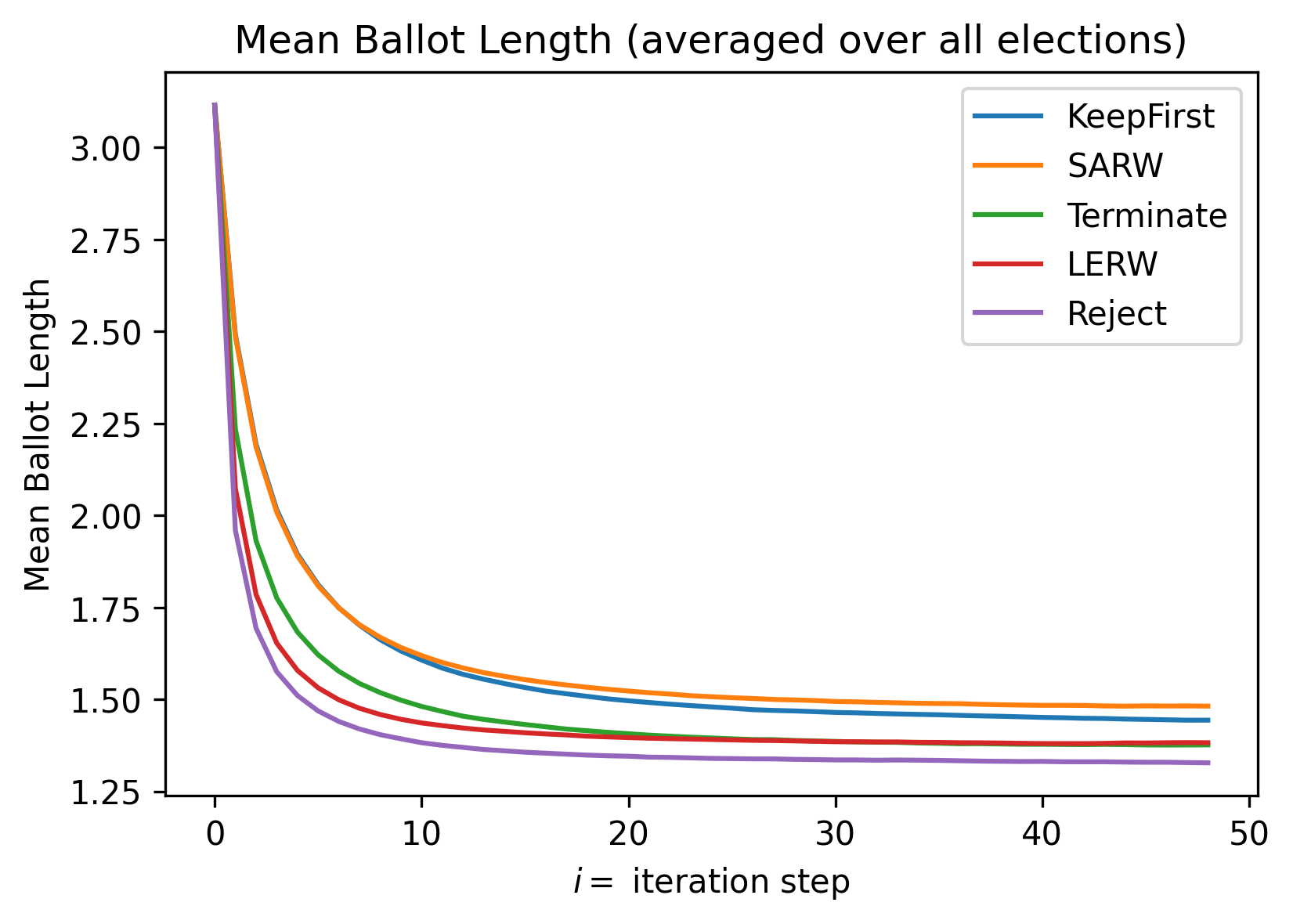}
    \caption{The mean ballot length is also given as an average over all elections.}
    \label{F:loop3}
\end{figure}

Figure~\ref{F:loop3}  shows the average ballot length decreases through the iteration steps towards a limiting value that depends on the method. 

\end{document}